\documentclass[11pt]{article}
\usepackage[utf8]{inputenc}
\usepackage[savepaper,secthm,nofancy,authblk,nonatbib]{anishs}
\usepackage{xcolor}

\newcommand{\ep}{\varepsilon}

\title{Comb inequalities for typical Euclidean TSP instances}
\author[1]{Wesley Pegden\thanks{Research Supported in part by NSF grant DMS1700365}}
\author[2]{Anish Sevekari}
\affil[1,2]{Department of Mathematical Sciences, Carnegie Mellon University}
\affil[1]{wes@math.cmu.edu}
\affil[2]{asevekar@andrew.cmu.edu}

\newcommand{\tsp}{\mathtt{TSP}}
\newcommand{\comb}{\mathtt{Comb}}
\newcommand{\hk}{\mathtt{HK}}

\newcommand{\XXX}{\mathcal{X}}
\newcommand{\YYY}{\mathcal{Y}}

\begin{document}
\maketitle
\begin{abstract}
    We prove that even in the average case, the Euclidean Traveling Salesman Problem exhibits an integrality gap of $(1+\ep)$ for $\ep>0$ when the Held-Karp Linear Programming relaxation is augmented by all comb inequalities of bounded size.  This implies that large classes of branch-and-cut algorithms take exponential time for the Euclidean TSP, even on random inputs.
\end{abstract}

\section{Introduction}
The Euclidean Traveling Salesman problem is a rare case where complexity theory and computational practice can both be said to have outpaced the other.

On the theory side: Papadimitriou showed that the Euclidean TSP is NP-hard, while Arora \cite{arora1996polynomial} and Mitchell \cite{mitchell1999guillotine} described polynomial-time approximation schemes (PTAS) for the Euclidean TSP.  On the computational side: efficient implementations of these PTASs have not materialized to supplant the use of heuristics without provable guarantees, while on the other hand, branch-and-cut methods using these heuristics with LP-based lower bounds nevertheless have found (provably) optimal tours in random or real-world (rather than worst-case) problem instances of large size; the current record is a problem instance from an application to integrated circuit design with 85,900 ``cities'' \cite{applegate2006traveling}.

Underpinning the tension in these developments is the unresolved status (even subject to standard complexity assumptions) of the hardness of finding optimal tours on \emph{typical}---rather than worst-case---instances of the Euclidean TSP:

\begin{ques}
	Is there a polynomial-time algorithm for the Euclidean TSP which, given a collection of $n$ independent rand:om points, returns an optimal tour with probability $p_n$ where $p_n\to 1$ as $n\to \infty$?
\end{ques}

\subsection{Branch-and-cut for the Euclidean TSP}

One of the most successful computational approaches in practice to find optimal tours for the Euclidean TSP is the \emph{branch-and-cut} approach, discussed by Applegate, Bixby, Chv\'atal and Cook \cite{applegate2006traveling}, and implemented in Cook's software package \emph{Concorde}.

Before discussing branch-and-cut, let us first recall that the more general \emph{branch-and-bound} approach is a combinatorial optimization paradigm based on pruning a branched exhaustive search.  In the context of finding optimal TSP tours, the approach combines (sub-optimal) algorithms for finding tours subject to restrictions (e.g., edge inclusions/exclusions), methods to establish lower bounds on tour lengths subject to restrictions, and a branching strategy which recursively partitions the exhaustive search space into complementary sets of restrictions.  Efficiency of the approach depends on lower bound methods being strong enough on restricted instances to match the global performance of upper bound (tour-finding) approaches to quickly prune large parts of the search space.

Within this paradigm, \emph{branch-and-cut} algorithms for the TSP specialize by using an LP relaxation lower bound for the TSP, which, for each constrained instance, can be augmented by an adaptive choice of cutting planes.  The algorithm \emph{branches}, partitioning a problem instance into a collection of problem instances with complementary restrictions, and then prunes by searching for \emph{cut}ting planes for each.

Frieze and Pegden \cite{frieze2015separating} showed that regardless of the tour-finding algorithm used for upper bounds (i.e., even if it actually finds optimal tours), the branch and bound decision tree will inevitably have exponential size if lower bounds are found via the Held-Karp LP-relaxation of the TSP, without any additional cutting planes \cite{held1971traveling}.

This \emph{Held-Karp lower bound} on the tour is defined by the linear program:
\[
	\arraycolsep=1.4pt\def\arraystretch{2.2}
	\begin{array}{crrl}
		\multicolumn{4}{c}{\displaystyle \min \sum_{\{i,j\} \subseteq V} c_{\{ij\}}x_{\{ij\}}}                                          \\
		\multicolumn{4}{c}{\mbox{ subject to }}                                                                                         \\
		\hline
		\textsc{(i)} \qquad   & (\forall i)                            & \displaystyle \sum_{j\neq i} x_{\{ij\}}           & =2         \\
		\textsc{(ii)} \qquad  & (\forall \emptyset \neq S\subsetneq V) & \displaystyle\sum_{\{i,j\} \subseteq S}x_{\{ij\}} & \leq |S|-1 \\
		\textsc{(iii)} \qquad & (\forall i<j\in V)                     & \displaystyle x_{\{ij\}}                          & \in [0,1]  \\
	\end{array} \labeleqn{eqn:hk}.
\]
Let $\hk(X)$ denote the value of this LP on a set $X$.  Note that under assumption \textsc{(i)} in \eqref{eqn:hk}, \textsc{(ii)} can be replaced by
\[ (\forall \emptyset \neq S \subsetneq V) \, \sum_{i \in S, j \notin S} x_{\set{ij}} \ge 2 \labeleqn{eqn:subtour_elim}\]
as shown in Section $58.5$ in \cite{schrijver2003combinatorial}; these are known as \emph{subtour-elimination} constraints.


The branch-and-cut approaches used to solve TSP instances of significant size go beyond the branch-and-bound framework considered by Frieze and Pegden, by using additional cutting planes to further prune the TSP search space.  Perhaps the most important class of such cutting planes are the so-called \emph{comb-inequalities} (which are valid for any solution $x$ corresponding to a TSP tour \cite{grotschel1986clique}).

\begin{definition}[Comb Inequality]
	\label{def:comb_inequality}
	Given sets $H$ and $T_1, \ldots, T_t$ for odd $t$, such that $T_i \cap T_j = \emptyset$ and
	$T_i \cap H \neq \emptyset$, the \emph{comb inequality} associated to these sets is given by
	\[ \sum_{\substack{i \in H \\ j \notin H}} x_{\set{ij}} + \sum_{k=1}^t \sum_{\substack{i \in T_k \\ j \notin T_k}} x_{\set{ij}} \ge 3t + 1 \]
	In this case, we call $H$ to be the \emph{handle} and $T_i$ to be the \emph{teeth} of comb inequality. We refer to $C = H \cup \brac[\big]{\cup_{k=0}^t T_k}$ as the comb and we will use the term \emph{size of the comb} to denote $\abs{C}$.
\end{definition}
\noindent

We will obtain in this paper a proof that polynomial-time branch-and-cut algorithms based on comb inequalities of bounded size cannot solve the Euclidean TSP on typical instances. In particular, let $\comb_c(X)$ denote the value of the LP obtained by adding all comb
inequalities with combs of size at most $c$ to the Held-Karp LP relaxation of TSP.  For a random set $\XXX_n$ of $n$ points in $[0,1]^d$, we prove:
\begin{theorem}
	\label{thm:bnb_algo_bound}
	Suppose that we use branch and bound to solve the TSP on $\XXX_n$, using $\comb_c$ as a lower bound for some fixed constant $c$. Then w.h.p, the algorithm runs in time $e^{\Omega(n/\polylog(n))}$.
\end{theorem}
Note that this gives a almost-exponential lower bound on the runtime of any branch and bound strategy. Further, we have a slightly more general version of this result when $c$ is not a constant, but with a slightly weaker but still super-polynomial lower bound on the runtime:
\begin{theorem}
	\label{thm:bnb_algo_bound_general}
	Suppose that we use branch and bound to solve the TSP on $\XXX_n$, using $\comb_c$ as a lower bound for $c = O\brac{\tfrac{\log n}{\log \log n}}$. Then w.h.p, the algorithm runs in time $e^{\Omega(n^{0.5})}$.
\end{theorem}
The set of all combs of size $\tfrac{\log n}{\log \log n}$ has size at least $n^{\Omega(\log n/(\log \log n))}$, which is super polynomial. It is not clear that there should be a polynomial-time separation algorithm for this set of comb-inequalities. The known results for separation of comb inequalities are for combs with a bounded number of teeth \cite{carr1997separation}, and combs that are derived in a specific way \cite{carr2004separation}.

\noindent The proof of the two theorems above theorem, along with a precise definition of the branch-and-bound paradigm we consider, can be found in \Cref{sec:bnb_algo_bound}.  The applicability of \Cref{thm:bnb_algo_bound,thm:bnb_algo_bound_general} to branch-and-cut follows from the fact that a branch-and-cut tree using only combs of size $\leq c$ contains as a subtree the corresponding branch-and-bound which uses $\comb_c$ as a lower bound.  The proofs of \Cref{thm:bnb_algo_bound,thm:bnb_algo_bound_general} depends on a new extension of probabilistic analyses of the Euclidean TSP and its LP relaxations.
\subsection{Probabilistic analysis of cutting planes for the Euclidean TSP}
The proof of Theorem \ref{thm:bnb_algo_bound} will depend on a probabilistic analysis of the impact of comb-inequality cutting planes on the value of the Held-Karp linear program \eqref{eqn:hk}.  In particular, if $x_1,x_2,\dots$ is a sequence of random points in $[0,1]^d$  and $\XXX_n=\{x_1,\dots,x_n\}$, we aim to show for any constant $c$ that for some $\ep>0$,
\[
	\comb_c(\XXX_n)\leq (1-\ep)\tsp(\XXX_n)\quad \mathrm{almost\, surely}\, (a.s.),
\]
where $\tsp(X)$ denotes the length of a shortest tour through $X$.  The random variable $\tsp(\XXX_n)$ was first studied by Beardwood, Halton and Hammersley \cite{beardwood1959shortest}. They proved in 1959 that there is an absolute constant $\beta^d_{TSP}$ such that the length
$\tsp(\XXX_n)$ of a minimum length TSP tour through $\XXX_n$ satisfies
\[ \tsp(\XXX_n) \sim \beta^d_{\tsp} n ^{\frac{d-1}{d}} \qquad a.s. \]
Here $a_n\sim b_n$ indicates that $a_n/b_n\to 1$.  This result has since been extended to many structures other than Hamiltonian cycles. Various similar results are also known for
problems like Minimum Spanning Tree \cite{beardwood1959shortest} and Maximum Matching
\cite{papadimitriou1978probabilistic}, etc. Steele \cite{steele1981subadditive} extended this result
to a more general framework which proves existence of such asymptotic constants $\beta_F$ for
\emph{subadditive Euclidean functional} $F$.
One peculiar feature of these results is that the true values of the constants are unknown, and even improvements on their estimates are rare.  Some results in this direction were proved in \cite{bertsimas1990asymptotic} and \cite{steinerberger2015new}.

Goemans and Bertsimas established in \cite{goemans1991probabilistic} an analogous asymptotic result for the Held-Karp linear program:
\[ \hk(\XXX_n) \sim \beta_{\hk}^d n^{\frac{d-1}{d}} \]
by proving that $\hk(X)$ is a subadditive Euclidean functional.  They asked in \cite{goemans1991probabilistic} whether $\beta_{\hk}^d = \beta_{\tsp}^d$; this was answered in the negative in the same paper \cite{frieze2015separating} showing that branch-and-bound with $\hk(\XXX_n)$ as a lower bound takes exponential time on typical inputs; Frieze and Pegden proved there that
\begin{equation}\label{eqn:FP}
	\beta_{\hk}^d < \beta_{\tsp}^d \qquad \forall d \ge 2.
\end{equation}

Let $\comb_c$ denote the value of the LP obtained by adding all comb
inequalities with combs of size at most $c$ to the Held-Karp LP relaxation of TSP. Since $\comb_c(X)
	\le \tsp(X)$ for all $x \in \mathbb{R}^d$, there is some constant $\gamma$ such that
\[ \limsup_{n \to \infty} \comb_c(\XXX_n) \cdot n^{-\frac{d}{d-1}} \le \gamma \quad \text{a.s.} \labeleqn{eqn:gamma_def} \]
Note that $\gamma = \beta_{\tsp}^d$ satisfies this equation.
\begin{definition}
	\label{def:gamma}
    Let $\Gamma$ denote the set of constants that satisfy \eqref{eqn:gamma_def}. Define
	\[ \gamma_{\comb}^{c,d} = \inf_{\gamma \in \Gamma} \gamma \]
\end{definition}
\noindent We claim that $\gamma_{\comb}^{c,d} \in \Gamma$. This holds since for all $m$, we have
\[ \pr{\limsup_{n \to \infty} \comb_c(\XXX_n) \cdot n^{-\frac{d-1}{d}} > \gamma_{\comb}^{c,d} + \frac{1}{m}} = 0 \]
By taking a countable union of all these events, we get that 
\[ \pr{\limsup_{n \to \infty} \comb_c(\XXX_n) \cdot n^{-\frac{d-1}{d}} > \gamma_{\comb}^{c,d}} = 0 \]
Proving that $\gamma_{\comb}^{c,d}$ satisfies \eqref{eqn:gamma_def} and lies in $\Gamma$.
With these definitions above, we will prove

\begin{theorem} For all constants $c$ and for all $d \ge 2$,
	\label{thm:comb_sep}
	\begin{equation}
		\label{eqn:comb_sep}
		\gamma_\comb^{c,d} < \beta_\tsp^d
	\end{equation}
\end{theorem}

\noindent The proof of \Cref{thm:comb_sep} appears in \Cref{sec:comb_ineqalities_separation}. In \Cref{sec:bnb_algo_bound} we show that this theorem implies \Cref{thm:bnb_algo_bound}.

\subsection{Notation}
Given a graph $G = (V,E)$ and $A,B \subset H$, let $\delta(A)$ denote the set of edges of $G$, with
exactly one vertex inside $A$. If $A,B$ are disjoint, then let $e(A,B)$ denote the set of edges
in $G$ with exactly one vertex in $A$ and one vertex in $B$.

\noindent A weight assignment $x$ is a function $x: E \mapsto \R$. Let $F \subset E$, then
\[ x(F) = \sum_{e \in F} x(e) \]
denotes the total weight of edges in $F$. In particular, $x(\delta(A))$ denotes the total weight leaving the set $A$, and $x(e(A,B))$ denotes the total weight of edges going from $A$ to $B$.

\section{Separating Constant Size Comb LP from TSP}
\label{sec:comb_ineqalities_separation}
Frieze and Pegden show in \cite{frieze2015separating} that for all $d \ge 2$, $\beta_\hk^d < \beta_\tsp^d$. They prove the result by constructing a gadget such that the length of any tour while passing through the gadget is significantly larger than the total contribution of a solution satisfying subtour elimination constraints. They then prove that suitable approximations to this gadget occur frequently enough in random set to ensure that the an LP solution can be found of length $(1-\ep)\tsp(\XXX_n)$. We now define this gadget $S(k)$.

\begin{definition}
  \label{d:s_k}
  The gadget $S(k)$ consists of $2k$ equally spaced points on the circle of radius $4$ and $k$ equally spaced points on the circle of radius $1$, along with the points $(2,0)$ and $(-2,0)$, which we refer to as the \emph{gap vertices}.
\end{definition}
\begin{figure}[t]
	\label{fig:original_single_entry}
	\centering
	\includegraphics{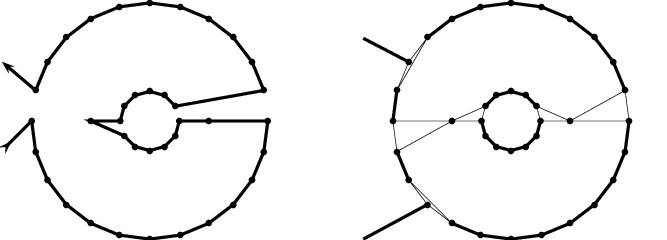}
	\caption{Solution when tour enters the gadget only once. Thick edges have value $1$ and thin edges have value $0.5$.}
\end{figure}

\begin{figure}[t]
	\label{fig:original_double_entry}
	\centering
	\includegraphics{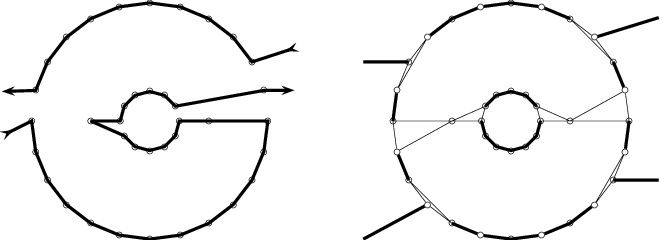}
	\caption{Solution when tour enters the gadget exactly twice. Thick edges have value $1$ and thin
	edges have value $0.5$.}
\end{figure}

\Cref{obs:gadget_entry_bound} (Observation $3.10$ from \cite{frieze2015separating}) states that we can enter a copy of this gadget at most twice.
\Cref{fig:original_single_entry} shows the gadget with a TSP (on the left) and corresponding Held-Karp solution (on the right) when the TSP enters/leaves the gadget just once, while \Cref{fig:original_double_entry} shows the same when the tour enters the same figure at most twice.
Note that in both the cases, the tour crossed the gap between smaller and larger circle roughly 3 times, while the half-integral solution (on the right) crosses this gap only twice (since the edges crossing the gap have weight $0.5$). Thus there is a constant gap between values of these solutions.

The proof in \cite{frieze2015separating} of \eqref{eqn:FP} incorporates the following two observations. Before stating them, we will recall an important definition from \cite{frieze2015separating}:

\begin{definition}
	\label{def:ep_d_copy}
	Consider a set $X \subset \R^n$. A set $T \subset X$ is $(\ep,D)$-copy of $S \subset \R^n$ if there is a set $S' \cong S$\footnote{For $S', S \in \R^n$, we write $S' \cong S$ if there is an isometry of $\R^n$ that maps $S$ to $S'$.} and a bijection $f$ between $T$ and $S'$ such that for all $x \in T$, $\norm{x - f(x)} < \ep$, and such that $T$ is at distance $> D$ from $X \setminus T$.
    
\end{definition}

Note that when we refer to a scaled $(\ep,D)$-copy of $S$, by say a factor $t$, we mean a $(t\ep, tD)$-copy of $t \cdot S$.

\begin{observation}[Observation 3.10 from \cite{frieze2015separating}]
\label{obs:gadget_entry_bound}
	Suppose that $S_{\ep,D}$ is an $(\ep,D)$ copy of any fixed set $S$ for fixed $\ep$
	and sufficiently large $D$. Then there are at most $2$ pairs of edges in a shortest TSP tour
	which join $S_{\ep,D}$ to $V \setminus S_{\ep,D}$.
\end{observation}

\begin{observation}[Observation 3.1 from \cite{frieze2015separating}]
\label{obs:smiley_face_lemma}
	Let $\set{ Y_1, Y_2, \ldots }$ be a sequence of points drawn uniformly at random
	from $[0,t]^d$ and $\mathcal{Y}_n=\{Y_1,\dots,Y_n\}$, where $t = n^{1/d}$. Given any finite point set $S$, any $\ep > 0$, and any
	$D$, $\mathcal{Y}_n$ a.s. contains at least $C_{\ep,D}^S n$ $(\ep,D)$-copies of $S$,
	for some constant $C_{\ep,D}^S > 0$.
\end{observation}

The structure of the proof of \Cref{eqn:FP} from \cite{frieze2015separating} is than as follows:
\begin{enumerate}[(i)]
	\item For $\YYY_n=t\cdot \XXX_n,$ \Cref{obs:smiley_face_lemma} ensures that we can choose a large constant $D$ and a small constant $\ep>0$ and find linearly many $(\ep,D)$-copies of the gadget described above.
	\item By \Cref{obs:gadget_entry_bound}, for each $(\ep,D)$ copy of the gadget, the shortest tour through $\YYY_n$ has either one or two components when restricted to the gadget.
	\item For both of these two possible cases, in each approximate copy of the gadget, the tour can be locally shortened by relaxing to a (half-integral) LP solution as in \Cref{fig:single_entry}.
	\item In total these shorten the tour by $\delta\cdot n$ for some $\delta>0$, which establishes \eqref{eqn:FP} since after rescaling by the factor $t$ we have that $\tsp(\YYY_n)\sim \beta^d_\tsp n.$
\end{enumerate}

To extend this approach to prove \Cref{thm:comb_sep}, we will do the following:
\begin{enumerate}[(1)]
	\item Construct a local half-integral solution on $S = S(k)$ assuming that tour visits $S$ exactly onces, entering and exiting through adjacent vertices (satisfying \Cref{a:adjecent_entry}).
	\item Prove that this solution satisfies all comb inequalities of size $c$ for $k = O(c)$ .
	\item Construct a gadget $\Pi^3(S)$ that contains $12$ copies of $S$ and any optimal tour through $\Pi^3(S)$ must go through at least one copy of $S$ while satisfying \Cref{a:adjecent_entry}.
\end{enumerate}

\noindent To begin, we prove some structural lemmas about combs.

\subsection{Technical lemmas for comb inequalities}
For the lemmas in this section, we suppose that $x$ is a half-integral solution to the Held-Karp LP, which has the property that all the edges of weight $1/2$ in $x$ form a graph that can be written as a union of edge disjoint triangles.
\begin{lemma}
\label{l:invalid_comb_struct}
	If $C$ is a comb violated by $x$ with handle $H$ and teeth $T_i$ for
	$i = 1 \ldots t$ for odd $t$, then following must hold:
	\begin{enumerate}[nosep]
		\item $x(\delta(H)) = t$
		\item $x(\delta^*(H)) = 0$
		\item $x(\delta(T_i)) = 2$ for all $i$.
		\item $x(e(A_i,B_i)) = 1$.
		\item $x(e(A_i,H \setminus A_i)) = 1$.
		\item $x(e(B_i,X \setminus (H \cup T_i))) = 1$.
	\end{enumerate}
	where $A_i = T_i \cap H$, $B_i = T_i \setminus H$ and $\delta^*(H)$ denotes the edges with
	exactly one endpoint inside $H$, and at least one endpoint outside $\bigcup_{i=1}^t T_i$.
\end{lemma}
\begin{proof}
Suppose $x$ violates the comb inequality $C$ with handle $H$ and teeth $T_1, \ldots, T_t$
for odd $t$.

Since this is a comb inequality, we know that $T_i$ intersect $H$, and are pairwise disjoint.
For each $i$, define $A_i = T_i \cap H$ and $B_i = T_i \setminus H$. For any set two sets $S, T$,
let $e(S,T)$ denote the set of edges with one endpoint in $S$ and another in $T$, and let
$\delta(S)$ denote the set of all edges with exactly one endpoint in $S$. Let $x$ denote 
the solution of LP that we are considering. That is, for any edge $e$, $x(e)$ denote the value 
associated to that edge. For any set $U \subseteq E$,
\[ x(U) = \sum_{e \in U} x(e) \]
is the total weight of the set of edges.

The comb-inequality constraint is given by
\[ x(\delta(H)) + \sum_{i=1}^t x(\delta(T_i)) \ge 3t + 1 \]
Since the comb inequality is not valid for the solution, we have
\[ x(\delta(H)) + \sum_{i=1}^t x(\delta(T_i)) < 3t + 1 \]

From subtour elimination, we have $x(\delta(A_i)) \ge 2$ and $x(\delta(B_i)) \ge 2$. Since $A_i$
and $B_i$ partition $T_i$, we have
\[ x(\delta(T_i)) = x(\delta(A_i)) + x(\delta(B_i)) - 2 x(e(A_i,B_i)) \labeleqn{eqn:cut_split} \]
Let $\delta^*(H)$ denote all the edges exiting $H$ that have are not contained inside a single tooth.
\[ \delta^*(H) = \delta(H) \setminus \left( \bigcup_{i=1}^t e(A_i,B_i) \right) \]
Substituting this into the comb inequality,
\[ x(\delta^*(H)) + \sum_{i=1}^t\big( x(e(A_i,B_i)) + x(\delta(T_i))\big) < 3t + 1 \labeleqn{eqn:comb_ineq} \]
Because of subtour elimination constraints, we have $x(\delta(T_i)) \ge 2$ for all $i$, which gives
\begin{align*}
	x(\delta^*(H)) + \sum_{i=1}^t x(e(A_i,B_i)) & < t + 1\\
	\implies \sum_{i=1}^t x(e(A_i,B_i)) & < t + 1 - x(\delta^*(H)) \labeleqn{eqn:sum_upper_bound}
\end{align*}
on the other hand, \eqref{eqn:cut_split} gives
\[ 
x(\delta^*(H)) + \sum_{i=1}^t \bigg(x(e(A_i,B_i)) + x(\delta(A_i)) + x(\delta(B_i)) - 2 x(e(A_i,B_i))\bigg) < 3t + 1 \]
\begin{align*}
	\implies \sum_{i=1}^t \bigg( x(\delta(A_i)) + x(\delta(B_i))\bigg) - 3t - 1 + x (\delta^*(H)) & < \sum_{i=1}^t x(e(A_i,B_i)) \\
	\implies t - 1 + x(\delta^*(H)) & < \sum_{i=1}^t x(e(A_i,B_i)) \labeleqn{eqn:sum_lower_bound}
\end{align*}
Combining the both, we have
\[ t - 1 + x(\delta^*(H)) < \sum_{i=1}^t x(e(A_i,B_i)) < t + 1 - x(\delta^*(H)) \labeleqn{eqn:sum_bound} \]
This immediately forces only two possible values of $x(\delta^*(H))$, either $0$ or $1/2$.

Substituting the lower bound \eqref{eqn:sum_lower_bound} into \eqref{eqn:comb_ineq}, we get
\begin{align*}
	x(\delta^*(H)) + \sum_{i=1}^t x(\delta(T_i)) + t - 1 + x(\delta^*(H)) & < 3t + 1 \\
	\implies \sum_{i=1}^t \left(x(\delta(T_i)) - 2\right) & < 2 - 2 x(\delta^*(H))
\end{align*}

Recall that edges of weight $1/2$ in $x$ form a graph that can be written as union of edge 
disjoint triangles. For any set $S$, any triangle can have either exactly $2$ edges crossing it,
or no edges crossing it. Hence, for any $S$, a triangle with all edges of weight $1/2$ contributes
either $1$ or $0$ to $x(\delta(S))$. Since we can decompose all the edges of weight $1/2$ into
edge disjoint triangles, no edges are double counted while adding up elements in $\delta(S)$, so 
for each set $S$, $x(\delta(S))$ is an integer. Now, observing the equation above, we can note that 
there is at most one $T_i$ for which $x(\delta(T_i)) = 3$. Further, even this cannot happen if
$x(\delta^*(H)) = 1/2$.

Now we are left with three cases, namely:
\begin{enumerate}[nosep]
	\item $x(\delta^*(H)) = 1/2$ and $x(\delta(T_i)) = 2$ for all $i$.
	\item $x(\delta^*(H)) = 0$, $x(\delta(T_1)) = 3$ and $x(\delta(T_i)) = 2$ for all $i \neq 1$.
	\item $x(\delta^*(H)) = 0$ and $x(\delta(T_i)) = 2$ for all $i$.
\end{enumerate}
we will show that only case $(3)$ can happen.

\case In this case,
\[ x(\delta(T_i)) = x(\delta(A_i)) + x(\delta(B_i)) - 2 x(e(A_i,B_i)) = 2 \]
since $x(\delta(A_i)), x(\delta(B_i)) \ge 2$, this gives $x(e(A_i,B_i)) \ge 1$. Substituting this
in \eqref{eqn:sum_upper_bound} gives
\[ t + 1 - \frac{1}{2} > \sum_{i=1}^n x(e(A_i,B_i)) \ge t \]
Since the sum only takes half integral values, this forces the value of the sum to be $t$.
So, $x(e(A_i,B_i)) = 1$ for all $i$. Now, note that
\[ x(\delta^*(H)) = x(\delta(H)) - \sum_{i=1}^t x(e(A_i,B_i)) = x(\delta(H)) - t \]
which implies that $x(\delta^*(H))$ is an integer since $x(\delta(H))$ is an integer, and hence
must be zero, forcing us to be in \Cref{c:invalid_comb_struct} instead.

\case In this case, by the same argument as in case $(1)$, we have 
$x(e(A_i,B_i))$ $\ge 1$ for all $i > 1$, and $x(e(A_1,B_1)) \ge 1/2$. Substituting these
values in \eqref{eqn:comb_ineq} gives
\[ t - \frac{1}{2} \le \sum_{i=1}^t x(e(A_i,B_i)) < 3t + 1 - \sum_{i=1}^t x(\delta(T_i)) = t \]
which forces equality on the left since summation only takes integer values. Therefore,
$x(e(A_1,B_1)) = 1/2$, and thus there is exactly one edge of weight $1/2$ between $A_1$ and $B_1$.
This edge is part of a triangle, whose vertex must lie outside $T_i$. But, then it contributes to
$x(\delta^*(H))$, and will contradict the assumption that $x(\delta^*(H)) = 0$. Therefore, case
$(2)$ can't hold either, which means we are in \Cref{c:invalid_comb_struct}

\case\label{c:invalid_comb_struct} Now we have $x(\delta^*(H)) = 0$, and hence if there is an edge of weight
$1/2$ in $e(A_i,B_i)$, then the unique triangle containing that edge in the decomposition must also
be completely contained in $T_i$. Therefore, every triangle with edges of weight $1/2$ contributes
either $1$ or $0$ to $x(e(A_i,B_i))$. Further, by the same argument as in analysis in case $(1)$,
$x(e(A_i,B_i)) \ge 1$ for all $i$, and using \eqref{eqn:sum_upper_bound} gives
\[ t + 1 > \sum_{i=1}^t x(e(A_i,B_i)) \ge t \]
Forcing the following equalities for all $i$:
\begin{enumerate}[nosep]
	\item $x(e(A_i,B_i)) = 1$.
	\item $x(e(A_i,H \setminus A_i)) = 1$.
	\item $x(e(B_i,X \setminus (H \cup T_i))) = 1$.
\end{enumerate}
and these are the only non empty boundary crossings with respect to $x$ for $A_i,B_i$. Thus, each
of these boundaries is either an edge of weight $1$, or a triangle with two edges of weight $1/2$
crossing the boundary. This completes the proof.
\end{proof}

Now we are ready prove a couple of trivial lemmas. But first, we will define \emph{induced subgraphs} with respect to an assignment $x$.

\begin{definition}
	\label{def:induced_subgraph}
	Given a set $X \in \R^n$, and an assignment $x$, for every subset $Y \subseteq X$, we define $G[Y]$ to be the graph with vertex set $Y$ and edges $e$ with both endpoints in $Y$ such that $x(e) > 0$.
\end{definition}

\begin{lemma}For any comb $C$ violated by $x$, with teeth $T_i$, the induced subgraph $G[T_i]$ is connected for all teeth $T_i$. 
\end{lemma}
\begin{proof}
	Suppose not, then applying subtour elimination constraint on each connected component (there are at least two) gives $x(\delta(T_i)) \ge 4$.
\end{proof}
\begin{lemma}
	\label{lemma:connected_handle}
	Consider an assignment $x$ and a comb $C$ with handle $H$ and teeth $T_i$ such that $x$ violates the comb inequality corresponding to the comb $C$. If $C$ is the comb with least number of teeth such that $x$ violates $C$, then the induced subgraph $G[H]$ is connected.
\end{lemma}
\begin{proof}
	Suppose not, and let $H_i$ be the connected components of $G[H]$. Let $\alpha_i = \set{j : T_j 
	\cap H_j \neq \emptyset}$ denote the set of teeth intersecting $H_i$. Note that by
	\Cref{l:invalid_comb_struct}, edges exiting any teeth into the handle must have weight $1$.
	This and the fact that weight $1/2$ edges form a graph that can be decomposed into edge disjoint
	triangles imply that a tooth can't interest two different connected components of the handle.
	Then, by the constraints in	\Cref{l:invalid_comb_struct}, $x(\delta(H_i)) = \abs{\alpha_i}$
	since edges in $H_i$ can only exit through some teeth $T_j$ with $j \in \alpha_i$, and they must
	exit with weight $1$. Therefore, it follows that
	\[ x(\delta(H_i)) + \sum_{j \in \alpha_i} x(\delta(T_i)) = 3 \abs{\alpha_i} \]
	Since at least one of the $\alpha_i$ must be odd, this gives us a smaller comb on which the
	solution violates the comb inequality, contradicting minimality of $H$.
\end{proof}
\begin{lemma}
	Any comb violated by $x$ must contain an edge of weight $1/2$ inside it.
\end{lemma}
\begin{proof} Suppose not.
	Note that edges exiting the handle exit through a tooth, so all of them must have weight $1$ by \Cref{l:invalid_comb_struct}.
	Since all the edges intersecting the handle have weight $1$, we can split the handle into connected
	components, which are paths. Note that each path contributes exactly $2$ to $x(\delta(H))$, and
	thus $x(\delta(H))$ must be even, which contradicts that $x(\delta(H)) = t$ is odd.
\end{proof}

\begin{definition}
	For any set $S$, define $E(S,n)$ to be the size of the smallest set $T \supseteq S$ such that
	$x(\delta(T)) \le n$.
\end{definition}
Note that $x(\delta(T_i)) = 2$ and $x(\delta(H)) = t$. Hence, a handle can only contain sets that
have small $E(S,t)$ values and a tooth can only contain sets that have small $E(S,2)$ values.

\begin{lemma}
	\label{lemma:cycle_lemma}
	Let $S \subset T$ be sets such that for all $u \in S$, $x(e(u, T \setminus S)) \le 1$.
	Suppose $x(\delta(S)) = n$ and $x(\delta(T)) = n - 1$.
	Then there are two vertices $u,v \in S$ such that $T$ contains a path from $u$ to $v$ outside $S$.
\end{lemma}
\begin{proof}
	For each $u \in S$, define $P_u$ to be the set of vertices in $T \setminus S$ that are
	connected to $u$ using edges in $T$ but outside $S$. If $P_u \cap P_v \neq \emptyset$ for some $u \neq v$,
	then there is a path from $u$ to $v$ strictly contained in $T \setminus S$, and $P_u = P_v$.
	
	Suppose this doesn't happen. Then $P_u$ are disjoint for all $u$. Let
	\[ T^* = T \setminus \left( S \cup \bigcup_{u \in S} P_u \right) \]
	There are no edges between $T^*$ and $S$ by definition. We have the following:
	\[ x(\delta(T)) = x(\delta(S)) + x(\delta(T \setminus S)) - 2 x (e(S,T \setminus S)) \]
	$T^*$ along with $P_u$ form a partition of $T \setminus S$. Note that there are no edges
	between any of these parts by definition. Therefore,
	\[ x(\delta(T \setminus S)) = x(\delta(T^*)) + \sum_{u \in S} x(\delta(P_u)) \]
	and if $\set{u,v} \in e(S, T \setminus S)$ with $u \in S$, then $v \in P_u$ by definition.
	Therefore, 
	\[ x(e(S,T\setminus S)) = \sum_{u \in S} x(e(u,P_u)) \]
	Using these identities, we get
	\[ x(\delta(T)) = x(\delta(S)) + x(\delta(T^*)) + \sum_{u \in S}\big( x(\delta(P_u)) - 2 x(e(u,P_u))\big) 
	\labeleqn{eqn:delta_T_sum} \]
	Observe that
	\[ x(\delta(P_u \cup u)) = x(\delta(u)) + x(\delta(P_u)) - 2 x(e(u,P_u)) \]
	and therefore that 
	\begin{equation}
	\label{eqn:othereq}
	x(\delta(P_u)) - 2 x(e(u,P_u)) = x(\delta(P_u \cup u)) - x(\delta(u))
	\end{equation}
	Now, we claim that $x(\delta(P_u \cup u)) \ge x(\delta(u))$. We split each of the boundaries into
	two parts to get
	\[ x(\delta(u)) = x(e(u, X \setminus (P_u \cup u))) + x(e(u,P_u)) \]
	\[ x(\delta(P_u \cup u)) = x(e(u, X \setminus (P_u \cup u))) + x(e(P_u, X \setminus (P_u \cup u))) \] 
	Subtracting the equations, we get
	\[ x(\delta(P_u \cup u))) - x(\delta(u)) = x(e(P_u, X \setminus (P_u \cup u))) - x(e(P_u,u)) \]
	On the other hand,
	\[ x(e(P_u, X \setminus (P_u \cup u))) + x(e(P_u,u)) = x(\delta(P_u)) \ge 2 \]
	Now, the condition that $x(e(P_u,u)) \le x(e(u, X \setminus S)) \le 1$, it must be the case that
	$x(e(P_u, X \setminus (P_u \cup u))) \ge 1 \ge x(e(P_u,u))$. This implies that
	\[ x(\delta(P_u \cup u))) - x(\delta(u)) = x(e(P_u, X \setminus (P_u \cup u))) - x(e(P_u,u)) \ge 0 \]
	for all $u$. Substituting this into \eqref{eqn:delta_T_sum} (using \eqref{eqn:othereq}),
	\[ x(\delta(T)) \ge x(\delta(S)) + x(\delta(T^*)) \]
	which is clearly false, since $x(\delta(T)) < x(\delta(S))$ by assumption. This completes the proof
	of the lemma.
\end{proof}

\subsection{Construction of Half Integral Solution for the Gadget}
\label{ss:half_integral_solution_gadget}

\begin{figure}[t]
	\centering
	\includegraphics{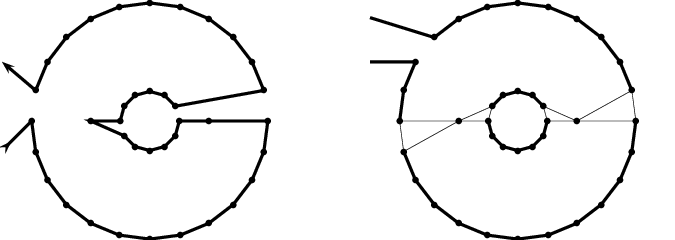}
	\caption{The gadget with the tour enters it once. Thick edges have weight $1$ and thin edges have weight $0.5$.}
    \label{fig:single_entry}
\end{figure}


We will now describe the construction of a local half-integral modification of the tour at an $(\varepsilon, D)$-copy of the gadget $S$, which is compatible with all comb inequalities of size $c$. For any Hamiltonian tour $P$, this modification can be made on any $(\varepsilon,D)$-copy of $S$ that has the following property with respect to $P$:

\begin{property}
	\label{a:adjecent_entry}
	We say that an $(\varepsilon, D)$-copy $S_1$ of $S$ has \Cref{a:adjecent_entry} with respect to a Hamiltonian path or tour $P$ if and only if 
	\begin{enumerate}[nosep]
        \item $P$ visits $S_1$ exactly once, and enters and leaves through consecutive vertices on the outer circle vertices.
		\item If $x,y$ are the points adjacent to $S_1$ in $P$, then the points $x,y$ are respectively connected to points of $S_1$ which are closest to them.
	\end{enumerate}
\end{property}

We will construct another gadget $\Pi^3_S(k)$ in \Cref{ss:expanding_the_gadget} that contains multiple copies of $S = S(k)$, such that given any optimal Hamiltonian tour $P$, at least one $(\varepsilon,D)$-copy of $S$ contained in an approximate copy of $\Pi^3_S$ must satisfy \Cref{a:adjecent_entry} with respect to $P$. 

The local half-integral solution mentioned above on $S$ (see \Cref{fig:single_entry}) consists of 
\begin{enumerate}[(i)]
    \item Four edge-disjoint triangles of edges of weight $\tfrac 1 2$---for each gap vertex, one such triangle joins that point to the closest two points on the outer and inner circles, respectively;
    \item Weight-1 edges joining the remaining consecutive pairs of points on the inner ring of the gadget;
    \item Weight-1 edges joining the remaining consecutive pairs of points on the outer ring of the gadget, except between entry/exit edges.
\end{enumerate}
Moreover, we require that the entry/exit edges are separated by at least $c-1$ points on the circle from the weight $\tfrac 1 2$ edges.
\Cref{fig:single_entry} shows the local solutions when $c = 2$ and $k = 12$, under \Cref{a:adjecent_entry}.
Now, we have the following lemma:

\begin{lemma}
	\label{l:s_k_bound}
	Consider gadget $S = S(k)$. Let $x,y$ be points outside $S$ and let $P$ be a Hamiltonian path $P$ from $x$ to $y$ in $\set{x,y} \cup S$ satisfies \Cref{a:adjecent_entry}. Then length of Hamiltonian path $P$ is at least 
	\[ \dist(x,S) + \dist(y,S) + 10 \pi + 8 - O\brac{\frac{1}{k}} \]
	On the other hand, cost of the half-integral solution described above is at most
	\[ \dist(x,S) + \dist(y,S) + 10 \pi + 6 + O\brac{\frac{c}{k}} \]
\end{lemma}
\begin{proof}
	Proof of the first lower bound is given in \cite{frieze2015separating}.
	We include a discussion about the lower bound in \Cref{ss:proof_s_k_bound} for sake of completeness.

	For the second bound, observe that in the half-integral solution, the total length of half-integral edges is $12 + 8\pi \frac{4}{2k} + 2\pi \frac{4}{k}$.
	On the other hand, the total length of integral edges contained in $S$ is $10 \pi - 8 \pi \frac{3}{2k} - 2 \pi \frac{2}{k}$, since we are missing $3$ edges on bigger circle and $2$ edges on smaller circle.
	Further, length of entry and exit segments is at most $\dist(x,S) + \dist(y,S) + 2 \frac{8\pi c}{2k}$ since these are the original entry / exit points moved by length at most $c$ points. Therefore, we get the total length of at most
	\[ \dist(x,S) + \dist(y,S) + \frac{8 \pi c}{k} + 10 \pi - \frac{16}{k} + 6 + \frac{12}{k} \]
	which gives the required bound.
\end{proof}

\begin{corollary}
	\label{c:gadget_separation}
	There exists a constant $\gamma$ such that if $k = \gamma c$ and $S = S(k)$, then for any points points $x,y$, and an Hamiltonian path from $x$ to $y$ on $S \cup \set{x,y}$ such that $S$ satisfies \Cref{a:adjecent_entry} with respect to $P$, the half-integral solution described above has total value at least $1$ smaller than length of $P$.
\end{corollary}
In particular, $\gamma = 16\pi$ ensures that for $c \ge 3$, $k \ge 48\pi$, the total cost of the half integral solution is at most $L + 6.5$. Following the computations in \Cref{ss:proof_s_k_bound}, total length of any Hamiltonian path $P$ from $x$ to $y$ on $S_1$ is at least $L + 7.5$, which implies the result.

\subsubsection{Satisfying combs with 3 teeth}
Now we prove that the half-integral solution described above in the gadget $S$ satisfies all $3$-combs of size at most $c$, assuming $c < k$. Note that we wil eventually pick $k = \gamma c$ where $\gamma > 2$, and hence this condition is always satisfied.

\begin{lemma}
	\label{lemma:large_3_comb}
	If a gap vertex is contained in a $3$-comb such that the comb inequality corresponding to the
	$3$-comb is violated, then $\abs{H \cup T_1 \cup T_2 \cup T_3} \ge 2c$.
\end{lemma}
\begin{proof}
	The gap vertex is contained in two triangles of weight-$\tfrac 1 2$ edges.

	\case If any triangle $P$ is contained in some tooth $T$, then note that $x(\delta(P)) = 3$ and $x(\delta(T)) = 2$. Further, each vertex of triangle has exactly $1$ weight going outside the triangle. Therefore, by \Cref{lemma:cycle_lemma}, $T$ contains a path between two vertices of the triangle that lies completely outside the triangle.
	Any such path either must go along entire inner circle, or entire outer circle or it exits the gadget and enters again. In first two cases, this path as length at least $2k$ and in second case, the path has length at least $2c$, implying that $\abs{T} \ge 2c$.

	\case If some tooth contains exactly two vertices of one of the two triangles, then since it doesn't contain the third vertex, all the conditions of \Cref{lemma:cycle_lemma} are satisfied. This again implies that $T$ contains a path between the two vertices that does not use any edges in the triangle, and hence must have size at least $2c$.

	Therefore, a tooth can contain at most one vertex of the triangles that contain the gap vertex. Hence, the handle must contain the gap vertex, and both the triangles containing the gap vertex. Let $Q$ denote the union of both the triangles. Then $x(\delta(Q)) = 4$, and every vertex has edges of weight exactly $1$ going out of $Q$. By \Cref{l:invalid_comb_struct} and the fact that this is a 3-comb, we have $x(\delta(H))=3$. 
	Thus, it satisfies the conditions of \Cref{lemma:cycle_lemma}, and hence, $H \setminus Q$ contains a path between two vertices in $Q$ which lies completely outside $Q$. This path must have length at least $2c$ by exactly the same argument as above.

	Hence, in all cases, either a tooth or the handle must contain at least $c$ vertices, proving the lemma that we want.
\end{proof}

Since all the half weight edges in the Gadget are contain at least one gap vertex, every $3$-comb must
have a gap vertex in it, and hence must have large size.

\subsubsection{Satisfying combs with 5 or more teeth}
\begin{lemma}
	\label{lemma:large_5_comb}
	If a gap vertex is contained in a $t$-comb, with $t \ge 5$, then 
	\[ \abs{H \cup \bigcup_{i=1}^t T_i} \ge c \]
\end{lemma}
\begin{proof}
	For this case, we will assume that the given comb in minimal, in particular, we have a comb with minimum value of $t$. If not, then we can always show the result for a smaller comb contained in this comb. From \Cref{lemma:connected_handle}, if $H$ is not connected, then we can always find a smaller comb that invalidates the solution. Hence, we will assume that $H$ is connected for rest of the proof.

	First, note that case $(1)$ of \Cref{lemma:large_3_comb} does not use the assumption that the comb has only $3$ teeth. Therefore, we can conclude that teeth of comb of any size cannot contain the gap vertex. Hence, we only need to handle the case that the handle of the comb contains a gap vertex.

	Note that any edge leaving the gadget is at least $c$ distance away from the gap vertex. Since the comb, that is $H \cup \bigcup_{i=1}^t T_i$ is connected, if the comb contains any vertex outside the gadget, then it must have at least $c$ vertices. Thus, we can assume that the comb is completely inside the gadget.

	Further, any tooth can't have an edge of weight $1/2$, since that would mean it contains two vertices of a gap triangle, and then by case $(1)$ of \Cref{lemma:large_3_comb} the tooth must contains a large cycle. Hence, all the edges strictly inside a tooth have weight exactly $1$. Hence, each tooth is completely contained inside one of the $4$ paths left after deleting all the edges of weight $1/2$ in the gadget. Note that there are only $4$ paths and at least $5$ teeth. Let $L_1, L_2, L_3, L_4$ be the $4$ paths.

	Hence, one of the paths, say $P$, contains at least $2$ teeth. Let these be $T_1, T_2$ such that the closest point in $T_1$ is closer to the gap vertex than the closest point in $T_2$. Now, since $T_1, T_2$ are connected, this implies that every vertex in $T_1$ is closer to the gap vertex than every vertex in $T_2$. But, since $T_2$ intersects the handle, there is a path from the gap vertex to $T_2$, say $Q$, which is completely contained in the handle. If $Q$ is completely contained in $P$, then it contains entire $T_1$ implying that $T_1$ is contained in the handle, which contradicts definition of the comb. Otherwise, the path $Q$ must wrap around using one of the other paths, $L_i \neq P$. Since it must contain the whole path, that implies that the handle has size at least $k$. This completes the proof.
\end{proof}

\subsection{Expanding the gadget}
\label{ss:expanding_the_gadget}

\begin{figure}[t]
	\centering
	\includegraphics{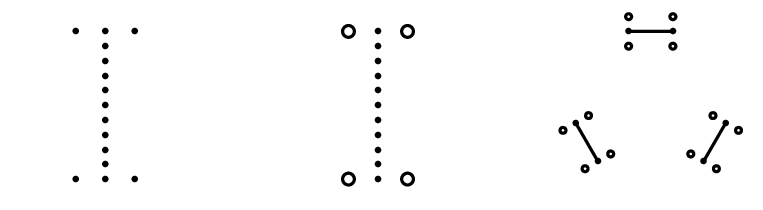}
	\caption{The figure shows (a)$\Pi(t,h,w)$, (b)$\Pi(S,t,h,w)$ and (c)$\Delta(S,D)$ from left to right.}
	\label{fig:pi_3_s}
\end{figure}

We will now describe the gadget $\Pi^3_S = \Pi^3_S(k)$ that contains $12$ copies of $S(k)$, such that there is an almost optimal Hamiltonian tour that satisfies \Cref{a:adjecent_entry} in at least one copy of $S$ in all $(\varepsilon,D)$ copies of $\Pi^3_S$. This gadget is obtained by combing two more gadgets with $S$, namely $\Pi_S = \Pi(S,t,h,w)$ and $\Pi^3_S = \Delta(\Pi_S,D)$. \Cref{fig:pi_3_s} illustrates these gadgets. We will provide the formal definitions and state few lemmas below, proofs of which are given in \Cref{s:gadget_proofs}.
\begin{restatable}[]{definition}{dPiTHW}
	\label{d:pi_t_h_w}
	We define the gadget $\Pi(t,h,w)$ for $t \in \Z_{\ge 0}$ and $h,w \in \R_{\ge 0}$, given by points $\pi_1 = \brac{-\frac{w}{2},0}$, $\pi_2 = \brac{\frac{w}{2},0}$, $\pi_3 = \brac{-\frac{w}{2},h}$, $\pi_4 = \brac{\frac{w}{2},h}$ and points $v_1, \ldots, v_t$ which are evenly spaced along $(0,0), (0,h)$, with $v_1 = \brac{0,0}$ and $v_t = \brac{0,h}$. 
	We will refer to sets $\set{\pi_1 \pi_2}$ and $\set{\pi_3 \pi_4}$ as \emph{shorter sides} of the gadget, and sets $\set{\pi_1 \pi_3}$ and $\set{\pi_2 \pi_4}$ as \emph{longer sides} of the gadget.
\end{restatable}
\begin{restatable}[]{lemma}{lPiNiceCorner}
	\label{l:pi_nice_corner}
	Let $p,q$ be two points on the opposite sides of the horizontal line $y=\tfrac h 2$ such that 
	\[ \dist(\set{x,y}, \Pi(t,h,w)) \ge D \]
	Let $P$ be a shortest Hamiltonian path from $p$ to $q$ in $\Pi(t,h,w) \cup \set{p,q}$. If all of the following inequalities hold,
	\[
		D \ge \tfrac{h^2 + w^2}{4w} \qquad
		h \ge 2w \qquad
		t \ge \tfrac{16h}{w}
	\]

	Then for at least two $i \in {1,2,3,4}$ we have that neither neighbor $v_i^1, v_i^2$ of $\pi_i$ on $P$ is not in $\set{p,q}$ and moreover, $v_i^1$, $v_i^2$ are two points in $\set{v_1, \ldots, v_t}$ closest to $\pi_i$.
\end{restatable}

Intuitively, this lemma holds since the shortest path through $\Pi(t,h,w)$ must travel through both the shorter sides, and connect them using the middle segment.
The condition on positions of $p,q$ ensures that it is beneficial to enter the gadget on one of the shorter sides and exit from the other shorter side. 
A formal proof is given in \Cref{ss:proof_pi_nice_corner}.

Now, we extend this gadget to the gadget $\Pi_S$, which is constructed by replacing each of the four corner points in $C = \set{\pi_1, \pi_2, \pi_3, \pi_4}$ by a copy of gadget $S(k)$ defined in \Cref{d:s_k}. 

\begin{restatable}[]{definition}{dPiS}
	\label{d:pi_s}
	We construct the gadget $\Pi(S(k),t,h,w)$ by replacing points in $C$ by copies of $S(k)$ centered at each point $\pi_i \in C$. We let $S_i$ denote the copy centered at $\pi_i$.
\end{restatable}


\begin{restatable}[]{lemma}{lPiSNiceCorner}
	\label{l:pi_s_nice_corner}
	Let $p,q$ be two points on the opposite sides of the line $y = \frac{h}{2}$ such that \
	\[ dist(\set{p,q}, \Pi(t,h,w)) \ge D. \]
	Let $P$ be a shortest Hamiltonian path from $p$ to $q$ in $\Pi(S(k),t,h,w) \cup \set{p,q}$. If all of the following inequalities hold,
	\[
		D \ge \tfrac{h^2 + w^2}{4w} \qquad
		h \ge 2w \qquad
		w \ge 100 \qquad
		t \ge 2h \qquad
		\frac{h}{t} \le \frac{4\pi}{k}
	\]

	Then there is a Hamiltonian path $Q$ from $p$ to $q$ in $\Pi(S(k),t,h,w) \cup \set{p,q}$ such that $Q$ visits each $S_i$ at most once, $\ell (Q) \le \ell (P) + O(1/k)$ and for at least two $i \in {1,2,3,4}$ we have that neither neighbor $v_i^1, v_i^2$ of $S_i$ on $Q$ is not in $\set{p,q}$ and moreover, $v_i^1$, $v_i^2$ are two points in $\set{v_1, \ldots, v_t}$ closest to $S_i$.
\end{restatable}

 In particular, $\Pi_S(k) = \Pi(S(k), \frac{200k}{4\pi}, 200, 100)$ satisfies this lemma for $D = 125$. A complete proof of \Cref{l:pi_s_nice_corner} is given in \Cref{ss:proof_pi_s_nice_corner}.
 
 Now, we introduce the final piece of the puzzle, the gadget $\Delta(D, \Pi_S(k))$ which contains three copies of the gadget $\Pi_S(k)$. This gadget is designed in a way that at least one of the copy of $\Pi_S(k)$ is visited exactly once in any optimal tour.

\begin{restatable}[]{definition}{dDelta}
	\label{d:delta}
	For any gadget $T \in \R^2$ with diameter $d$ and $D \in \R_{\ge 0}$, we define the gadget $\Delta(D,T)$ containing three copies of $T$, $T_1, T_2, T_3$ centered at points $C_1 = \brac{R, \tfrac{\pi}{2}}$, $C_2=\brac{R, \tfrac{7\pi}{6}}$ and $C_3=\brac{R,\tfrac{11\pi}{6}}$ for $R = \tfrac{D+2d}{\sqrt{3}}$ and rotated clockwise in angles of $\tfrac{\pi}{2}$, $-\tfrac{\pi}{6}$, and $\tfrac{\pi}{6}$ respectively.
\end{restatable}

An illustration for $\Delta(D,\Pi_S)$ is given in \Cref{fig:pi_3_s}. The rotations are to ensure that gadgets are symmetrically situated around rays along $OC_1$, $OC_2$ and $OC_3$ respectively. Further, distance between $T_i$ and $T_j$ is at least $D$ for any $i \neq j$. Now, we have following lemma about this gadget:

\begin{restatable}[]{lemma}{lDelta}
	\label{l:delta_single_visit}
	Let $\varepsilon > 0$ be positive real. Then there exists constants $D_1, D_2 \ge 0$ such that if $P$ is an optimal Hamiltonian tour over $V$, and if $\Delta_1$ is any $(\varepsilon,D_2)$ copy of $\Delta(D_1, \Pi_S(k))$, then there exists an $i \in \set{1,2,3}$ such that $P$ visits $\Pi_i$ exactly once, where $\Pi_1, \Pi_2, \Pi_3$ are $(\varepsilon, D_1)$-copies of $\Pi_S(k)$ contained in $\Delta_1$, with centers $C_1, C_2, C_3$ respectively. Further if $p,q$ are neighbors of $T_i$ in $P$, then $p,q$ lie on the opposite side of $\overleftrightarrow{OC_i}$, where $O$ is the center of $\Delta_1$. In particular, the values
	\begin{equation}
		\label{eq:delta_single_visit}
		D_1 = \frac{2000}{1 - \cos \frac{\pi}{10}} \qquad \text{and} \qquad D_2 = \frac{30000}{\brac{1 - \cos \frac{\pi}{10}}^2}
	\end{equation}
	suffice.
\end{restatable}

Proof of this lemma is a repeated application of \Cref{obs:gadget_entry_bound}, in particular, the condition that if an $(\varepsilon,D)$-copy of a gadget $T$ is visited multiple times, then the entry and exit rays must be parallel. The details are given in \Cref{ss:proof_delta_single_visit}.

We define $\Pi^3_S = \Pi^3_S(k) = \Delta(D_1, \Pi_s(k))$ where $D_1 = \tfrac{2000}{1 - \cos (\pi/10)}$ is as defined in \Cref{l:delta_single_visit}. Combining \Cref{l:pi_s_nice_corner,l:delta_single_visit}, we get the following lemma which shows the existence of $(\varepsilon,D)$-copies of $S$ which have \Cref{a:adjecent_entry}.

\begin{lemma}
	\label{l:nice_s_existance}
	For any $k \ge 4$ and $\varepsilon > 0$, there is are constants $D_1, D_2 > 0$ such that any $(\varepsilon,D_2)$-copy $\Pi^3_1$ of $\Pi^3_S(k)$ contained in $V$, and for any optimal Hamiltonian tour $P$ on $V$, there is a Hamiltonian tour $Q$ such that $\ell (Q) \le \ell (P) + 40 k \varepsilon$ and $(\varepsilon, D_1)$-copy $S_1$ of $S(k)$ such that $S_1 \subseteq \Pi_1$ and $S_1$ has \Cref{a:adjecent_entry} with respect to path $Q$.

	Further, if $k \ge \gamma c$ where $\gamma$ is given by \Cref{c:gadget_separation} and if $\bar{y}$ is the half-integral solution described in \Cref{ss:half_integral_solution_gadget}, then replacing $Q$ by $S_1$ gives an half-solution satisfying $\comb_c$ inequalities. This replacement can be made in all disjoint $(\varepsilon,D_2)$-copies of $\Pi^3_S(k)$ in $V$ simultaneously.
\end{lemma}
\begin{proof}
	We pick $D_1, D_2$ as defined in \Cref{l:delta_single_visit}, namely
	\[ D_1 = \frac{2000}{1 - \cos \frac{\pi}{10}} \qquad \text{and} \qquad D_2 = \frac{30000}{\brac{1 - \cos \frac{\pi}{10}}^2} \]
	Recall that $\Pi_S(k) = \Pi(S(k), \tfrac{200k}{4\pi}, 200, 100)$. It follows from \Cref{d:s_k,d:pi_s,d:delta} that $\Pi_S(k)$ has at most $40k$ points for $k \ge 4$.
	Since $\Pi^3_1$ is an $(\varepsilon,D_2)$-copy of $\Pi^3_S(k)$, there exists a translation $T_1$ of $\Pi^3_S(k)$ and a bijection $f:T_1 \to \Pi^3_1$ such that $\norm{x-f(x)} \le \varepsilon$.

	Using \Cref{l:delta_single_visit}, there is an $(\varepsilon,D_1)$-copy $\Pi_1$ of $\Pi_S(k)$ that is visited by $P$ exactly once. Further, if $p$ and $q$ are the points adjacent to $\Pi_1$ in $P$, then $p,q$ are on opposite side of $OC_1$ where $O$ is center of $T_1$ and $C_1$ is the center of $T_2 = f^{-1}(\Pi_1)$.
	Let $\bar{P}$ be an optimal Hamiltonian path from $p$ to $q$ in $T_1 \cup \set{p,q}$. Then by \Cref{l:pi_s_nice_corner}, there is an Hamiltonian path $\bar{Q}$ from $p$ to $q$ in $T_1 \cup \set{p,q}$ such that $\ell (\bar{Q}) \le \ell (\bar{P}) + O(1/k)$ and a copy $T_2$ of $S(k)$ such that $S(k)$ has \Cref{a:adjecent_entry} with respect to $\bar{Q}$.
	Let $Q$ be the Hamiltonian tour that equals $P$ outside $f(T_1) \cup \set{p,q}$ and $f(\bar{Q})$ inside $f(T_1) \cup \set{p,q}$. Then $f(T_2)$ has \Cref{a:adjecent_entry} with respect to $Q$. Since $\bar{P}$ was optimal on $T_1 \cup \set{p,q}$ and $T_1$ has at most $40k$ points, we must have $\ell (\bar{P}) \le \ell (P \cap (T_1 \cup \set{x,y})) + 40k\varepsilon$. It follows that $\ell (Q) \le \ell (P) + 40 k \varepsilon + O(1/k)$.
	
	Using \Cref{c:gadget_separation} for $k = \gamma c$, and $\varepsilon = O(\tfrac{1}{k})$, we can ensure that the cost of half-integral solution $\bar{y}$ is at least $1$ smaller than length of $Q$, and at least $0.5$ smaller than length of $P$. Further, \Cref{ss:half_integral_solution_gadget} implies that we can make this replacement in a single gadget without violating $\comb_c$ inequalities.

	If we do the replacement simultaneously in multiple disjoint $(\varepsilon,D_2)$-copies of $\Pi^3_S(k)$, then any comb of size at most $c$ containing an half-integral edge must be completely contained in an $(\varepsilon,0)$-copy $S(k)$. And hence again by \Cref{ss:half_integral_solution_gadget}, we do not violate any $\comb_c$ inequalities.
\end{proof}

Now, we are in a position to complete proof of \Cref{thm:comb_sep}. Recall that $\tsp(\mathcal{Y}_n) = \beta_\tsp^d n$ a.s. (since $\mathcal{Y}_n$ are random points in $[0,t]^d$). Further, by \Cref{obs:smiley_face_lemma}, we can find $C_{\ep,D_2}^\Pi n$ disjoint $(\varepsilon,D_2)$-copies of $\Pi^3_S(k)$ in $\mathcal{Y}_n$ w.h.p., for $k = \gamma c$, $\varepsilon = O(1/k)$ and $D_2 = O(1)$, for some constant $C_{\ep,D_2}^\Pi$. (To be more precise, we can choose $C_{\ep, D_2}^\Pi = \exp\brac{-O(c \log c)})$.

Using \Cref{l:nice_s_existance}, given any optimal tour in $\mathcal{Y}_n$, we can find a half-integral solution $\bar{y}$ on the edges of $\mathcal{Y}_n$ which is at least $\tfrac{1}{2} C_{\ep,D_2}^\Pi n$ smaller than the length of optimal tour.
Therefore, 
\[ \comb_c\brac{\mathcal{Y}_n} \le \tsp\brac{\mathcal{Y}_n} - K n = (\beta_\tsp^2 - K)n \]
a.s. for some constant $K > 0$, implying that $\gamma_\comb^2 < \beta_\tsp^2 - K$.
This completes proof of \textbf{\Cref{thm:comb_sep}} in two dimensions.

\subsection{Higher Dimensions}
Note that the construction above only works for $d = 2$. For higher dimensions, we construct gadget $T_d(k)$ which contains $5$ copies of $\Pi^3_S(k)$ which are at least $D_2$ distance apart from each other and lie in the same $2$-dimensional plane, where $D_2$ is as defined in \Cref{l:nice_s_existance}.

Since an optimal tour can enter $T_d(k)$ at most twice (\Cref{l:two_exits}), at least one of the $5$ copies of $\Pi^3_S(k)$ must only be connected to points in $T_d(k)$. This reduces problem to two dimensional case, and we can then use \Cref{l:nice_s_existance} to conclude higher dimensional version of \Cref{l:nice_s_existance}!

\begin{lemma}
	For any $k \ge 4$ and $\varepsilon > 0$, there is are constants $D_1, D_2 > 0$ such that any $(\varepsilon,D_2)$-copy $T_d$ of $T_d(k)$ contained in $V$, and for any optimal Hamiltonian tour $P$ on $V$, there is a Hamiltonian tour $Q$ such that $l(Q) \le l(P) + 200 k \varepsilon$ and $(\varepsilon, D_1)$-copy $S_1$ of $S(k)$ such that $S_1 \subseteq \Pi_1$ and $S_1$ has \Cref{a:adjecent_entry} with respect to path $Q$.

	Further, if $k \ge \gamma c$ where $\gamma$ is given by \Cref{c:gadget_separation} and if $\bar{y}$ is the half-integral solution described in \Cref{ss:half_integral_solution_gadget}, then replacing $Q$ by $S_1$ gives an half-solution satisfying $\comb_c$ inequalities. This replacement can be made in all disjoint $(\varepsilon,D_2)$-copies of $T_d(k)$ in $V$ simultaneously.
\end{lemma}

In particular, similar to argument in $2$ dimensional case, this gives us the separation in higher dimensions, namely
\[ \gamma_\comb^d \le \beta^d_\tsp - K \]
for some constant $K$.

\section{Branch and Bound Algorithms}
\label{sec:bnb_algo_bound}
In this section, we will prove \Cref{thm:bnb_algo_bound}.
For this section, we will assume that we are working in some fixed dimension $d$. Further, throughout this section, $O$ notation will hide constants dependent on $d$.

As considered here, a branch and bound algorithm depends on three choices:
\begin{enumerate}[label=(\arabic*),nosep]
	\item A choice of heuristic to find (not always optimal) TSP tours.
	\item A choice of lower bound for TSP (such as $\comb_c$ or $\hk$).
	\item A branching strategy.
\end{enumerate}
The result of a branch-and-bound approach is a branch-and-bound tree, which is a rooted tree such that to each vertex $v$ of this tree, we associate two sets $I_v$ and $O_v$ such that 
\begin{enumerate}[label=(\arabic*),nosep]
	\item When $v$ is the child of $u$, $I_v \supseteq I_u$ and $O_v \supseteq O_u$
	\item If $u$ has children $v_1, \ldots, v_k$, then we have $\Lambda_u= \bigcup_{i=1}^k \Lambda_{v_i}$, where $\Lambda_u$ denotes the set of TSP tours which  include all the edges in $I_u$ and exclude all the edges in $O_u$.
	\item The leaves of the (unpruned) branch and bound tree satisfy $|\Lambda_v| = 1$.
\end{enumerate}
For any node $v$ of the branching tree, let $b_v$ denote the value of the lower bound, which in our case is the value of $\comb_c$ under the additional constraints given by $I_v$ and $O_v$ (that is, the solution for $\comb_c$ must include all the edges in $I_v$ with weight $1$ and must exclude all the edges in $O_v$).
Let $B$ be the value of the tour given by our heuristic. For each vertex $v$, we find a tour using the some heuristic that includes the edges in $I_v$ and excludes the edges in $O_v$, and whenever we find a tour smaller than $B$, we update $B$. For every vertex $v$ such that $b_v \ge B$, we know that we have already found a tour as good as any in $\Lambda_v$, and we prune the tree at $v$.
The process ends when the set $L$ of leaves of the pruned tree satisfies $v \in L \Rightarrow b_v \ge B$. Note that such a tree in fact gives a proof that $B$ is an optimal tour.

Note that following any branching strategy to generate the tree will give us an optimal tour and proof of its optimality. For a branch and bound to be efficient, we want to prune the tree such that only polynomially many leaves remain.

We can now state a more precise version of \Cref{thm:bnb_algo_bound} as follows:
\begin{restatable}[\Cref{thm:bnb_algo_bound} restated]{theorem}{tbnb}
	\label{thm:bnb_restated}
	For any TSP heuristic, any branching strategy and a lower bound heuristic which is $\comb_c$ for some constant $c$, the pruned branch and bound tree will have $e^{\Omega(n/\log^5 n)}$ leaves w.h.p.
\end{restatable}

Further, we state a generalization of above result when $c$ is not a constant as follows:

\begin{restatable}[\Cref{thm:bnb_algo_bound_general} restated]{theorem}{tbnbg}
	\label{thm:bnb_restated_general}
	Given any $\varepsilon > 0$, For any TSP heuristic, any branching strategy and a lower bound heuristic which is $\comb_c$ for $c = O\brac{\tfrac{\varepsilon \log n}{\log \log n}}$, the pruned branch and bound tree will have $e^{\Omega(n^{1 - 6\varepsilon})}$ leaves w.h.p.
\end{restatable}
\noindent Note that setting $\varepsilon = 0.08$ gives us \Cref{thm:bnb_algo_bound_general}

Any branch-and-bound approach should produce not only an optimal tour, but, via the pruned tree and computed bounds, a certificate verifying that the returned tour is optimal.  \Cref{thm:bnb_restated} shows that even just the size of this certificate is exponential.
Our general strategy to prove \Cref{thm:bnb_restated} will be to show that when $\comb_c(\XXX_n \mid I_v, O_v) \ge \tsp(\XXX_n)$ then either $I_v$ or $O_v$ is must be large, and hence $\Lambda_v$ is in fact small. 
Since $\Lambda = \bigcup \Lambda_v$, this would imply that there are a lot of leaves in any pruned tree.

Following \cite{frieze2015separating}, we will further modify this approach by looking at a special set of tours $\overline{\Lambda}$. Given the point set $\mathcal{X}_n$, we will consider the division of $[0,1]^d$ into $s = \frac{n}{\sigma}$ boxes of side-length $s^{-\frac{1}{d}}$. We will eventually $\sigma = \Omega(\log n)$ as required for the runtime bounds.
Let $B_1, \ldots, B_s$ denote these boxes, taken in some order such that consecutive terms share a $(d-1)$ dimensional face. Note that
\[ \abs{x - y} \le \sqrt{d} \cdot s^{-\frac{1}{d}} = O\brac{s^{-\frac{1}{d}}} \]
if $x,y$ lie in the same box. We consider $\XXX_n = \set{x_1, \ldots, x_n}$, and for each $2 \le j \le s-1$, we let $x_j^1, x_j^2, x_j^3, x_j^4$ denote the four points $x_i \in \XXX_n \cap B_j$ of smallest index (this choice can be arbitrary, and is just for definiteness).
We also chose points $x_1^3, x_1^4 \in \XXX_n \cap B_1$ and $x_s^1, x_s^2 \in \XXX \cap B_s$, again by simply choosing points of minimum index. These points chosen as above can be viewed as preselected interface points between boxes $B_j$. In particular, we let $\overline{\Lambda}$ denote the set of TSP tours in $\XXX_n$ with the properties that, in that tour,
\begin{enumerate}[nosep]
	\item $x_1^4$ is joined to $x_1^3$ by a path lying entirely in $B_1$;
	\item for $1 \le j \le s-1$, $x_j^3$ and $x_{j+1}^1$ are adjacent;
	\item for $2 \le j \le s-1$, $x_j^1$ is joined to $x_j^3$ by a path lying completely in $B_j$;
	\item $x_s^1$ is joined to $x_s^2$ by a path lying entirely in $B_s$;
	\item for $s \ge j \ge 2$, $x_j^2$ and $x_{j-1}^4$ are adjacent; and
	\item for $s - 1 \ge j \ge 2$, $x_j^2$ is joined to $x_j^4$ by a path lying completely in $B_j$.
\end{enumerate}
\vspace*{1ex}
We will only restrict our attention to these special tours. Note that we are now only looking at
a smaller subset of tours. We claim that these tours have asymptotically
almost the same length as the TSP tour \emph{almost surely}. These tours are similar to those produced by Karp's \emph{fixed dissection} heuristic \cite{karp1977probabilistic}, which divides the square into $s$ boxes like we have, finds optimal tours through each, and then joins into a closed walk by means of an optimal tour through a set of representatives.

\newcommand{\tspbar}{\overline{\tsp}}
\newcommand{\tspf}{\tsp^F}

For the sake of notation, let $\tspbar(\XXX_n)$ denote the best tour in $\overline{\Lambda}$. Let $\tspf(\XXX_n)$ denote the tour given by fixed dissection heuristic. We claim that asymptotically
\[ \tsp(\XXX_n) \sim \tspbar(\XXX_n) \]
The proof is based off the techniques used to show $\tsp(\XXX_n) \sim \tspf(\XXX_n)$. We will leverage parts of Lemma 4 in Chapter 6 from \cite{lawler1985traveling}, in particular,
\begin{lemma}
	\label{l:karp_dissection}
	Let $\tsp(B_j)$ denote the best tour in $\XXX_n \cap B_j$. Then we have the following bound:
	\begin{equation}
		\label{eq:karp_dissection}
		\sum_{j=1}^s \tsp(B_j) \le \tsp(\XXX_n) + O\brac{n^{\frac{d-2}{d-1}}s^{\frac{1}{d(d-1)}}} + O\brac{s^{\frac{d-1}{d}}} = \tsp(\XXX_n) + O\brac{n^{\frac{d-1}{d}} \sigma^{-\frac{1}{d(d-1)}}} 
	\end{equation}
	where $s$ is the number of boxes $B_j$. Recall that $s = o(n)$.
\end{lemma}

Apart from finding the best tour in each cube $B_j$, the cost of modifying this solution into a path that starts at $x_j^1$ and ends at $x_j^3$ is at most $2 d^{1/2} s^{-1/d}$. The cost of patching edges between $B_j$ and $B_{j+1}$ also at most $2 d^{1/2} s^{-1/d}$. This gives us the upper bound:
\[ \tspbar(\XXX_n) \le \tsp(\XXX_n) + O\brac{n^{\frac{d-2}{d-1}}s^{\frac{1}{d(d-1)}}} + O\brac{s^{\frac{d-1}{d}}} + O\brac{s \cdot s^{-\frac{1}{d}}} \]
Since we choose $s = \tfrac{n}{\sigma} = o(n)$, we get
\begin{align*}
  \tsp(\XXX_n) 
	& \le \tspbar(\XXX_n) \le \tsp(\XXX_n) + O\brac{n^{\frac{d-1}{d}} \brac{\sigma^{-\frac{1}{d(d-1)}} + \sigma^{-\frac{d-1}{d}}}}\\
	& = \tsp(\XXX_n) + O\brac{n^{\frac{d-1}{d}} \sigma^{-\frac{1}{d(d-1)}}}
\end{align*}
where the $O$-notation hides constants dependent only on $d$. Note that this statement holds true deterministically.

Now we use the bounds on sizes of $\overline{\Lambda}$ and $\overline{\Lambda}_v = \Lambda_v \cap \overline{\Lambda}$ proved in \cite{frieze2015separating} (equations $29-32$). Let $\beta_j = \abs{\XXX_n \cap B_j}$, let $O_v^j$ denote the set of edges in $O_v$ that have both the endpoints in $B_j$ and let $I_v^j$ be the set of edges in $I_v$ that have both the endpoints in $B_j$. Let $I_v' \subseteq I_v$ denotes edges in $I_v$ of the form $\set*{x_j^3, x_{j+1}^1}$ or $\set*{x_j^2,x_{j-1}^4}$.
\begin{equation}
	\label{eqn:bound_l}
	\abs{\overline{\Lambda}} = (\beta_1 - 2)! \brac{\prod_{j=2}^{s-1} (\beta_j - 3)!} (\beta_s - 2)! 
\end{equation}
This bound follows since we can choose tour in every box $B_j$ by choosing the path from $x_j^1$ to $x_j^3$ and the path from $x_j^2$ to $x_j^4$, by choosing a permutation of $(\beta_j - 4)$ vertices ($(\beta_j - 4)!$ choices) and breaking it up into $2$ parts ($\beta_j - 3$ choices). First and last terms follow from a similar logic on box $B_1$ and $B_s$, which only have $2$ special vertices instead of $4$. Now, observe that $\Lambda_v = \emptyset$ unless $I_v = I_v' \cup \bigcup_{j=1}^s I_v^j$. To get an upper bound on $\overline{\Lambda}_v$, we look at the portion of tour in $B_j$, which can be represented as a permutation of $(\beta_j-3)$ symbols. Given an orientation of edges in $I_v^j$, each edge reduces the number of free symbols in the permutation by at $1$, giving us an upper bound of
\begin{equation*}
	\abs{\overline{\Lambda}_v}
	\le (\beta_1 - 2 - \abs{I_v^1})! 2^{\abs{I_v^1}} \brac{\prod_{j=2}^{s-1}(\beta_j - 3 - \abs{I_v^j})! 2^{\abs{I_v}}} (\beta_s - 2 - \abs{I_v^s})! 2^{\abs{I_v}} \\
\end{equation*}
Let $\overline{I}_v = \bigcup_{j=1}^s I_v^j$. Using Sterling's Approximation, we get
\begin{equation}
	\label{eqn:bound_l_v_inc}
	\abs{\overline{\Lambda}_v}
	\le \abs{\overline{\Lambda}} \cdot \prod_{j=1}^s \brac{\frac{2e}{\beta_j - 3}}^{\abs{I_v^j}}
	\le \abs{\overline{\Lambda}} \cdot e^{-\abs{\overline{I}_v}}
\end{equation}
assuming that $\beta_j \ge 2e^2 + 3$ for all $j$. Note that a Chernoff bound gives $\beta_j \in (1 \pm 0.5) \sigma$ for all $j$ with probability $n e^{-\sigma}$, where $s = \tfrac{n}{\sigma}$. This implies that \Cref{eqn:bound_l_v_inc} holds with high probability provided that $\sigma = \Omega(\log n)$.

On the other hand, observe that number of permutations on $\beta_j - 3$ symbols that avoid one particular edge is at most 
\[ (\beta_j - 3)! - (\beta_j - 4)! \le \brac{1 - \frac{1}{\beta_j}}(\beta_j - 3)! \]
simply by subtracting number of permutations that include this particular edge. Define $\delta_A = 1$ if $\abs{A} \ge 1$ and $\delta_A = 0$ otherwise. Then we have an upper bound
\begin{equation*}
	\abs{\overline{\Lambda}_v} \le \abs{\overline{\Lambda}} \cdot \prod_{j=1}^s \brac{1 - \frac{\delta_{O_v^j}}{\beta_j}}
\end{equation*}
Let $\overline{O}_v = \bigcup_{j = 1}^s \overline{O}_v^j$. Since $\beta_j \le 2 \sigma$ for all $j$ with high probability, there must be at least $\abs{\overline{O}_v^j} \brac{\tfrac{1}{2\sigma}}^2$ integers $j$ such that $\abs{O_v^j} \ge 1$. Therefore, we get the upper bound:
\begin{equation}
	\label{eqn:bound_l_v_exc}
	\overline{\Lambda}_v \le \overline{\Lambda} \cdot \brac{1 - \frac{1}{2\sigma}}^{\abs{\overline{O}_v}
	\brac{\frac{1}{2\sigma}}^2} \le \overline{\Lambda} \cdot e^{- \abs{\overline{O}_v}/(2\sigma)^3}
\end{equation}
Now that we have established these bounds, we know that a large $\overline{I}_v$ or large $\overline{O}_v$ forces $\overline{\Lambda}_v$ to be small. Define $\overline{L} = \set{v \in L \st \overline{\Lambda}_v \neq \emptyset}$. Note that $\overline{\Lambda} = \bigcup_{v \in \overline{L}} \overline{\Lambda}_v$. Now, since $\overline{\Lambda}$ itself is large, it suffices to show that $v \in \overline{L}$ implies that either $\overline{I}_v$ or $\overline{O}_v$ is large. Indeed, we have
\begin{lemma}
	Let $d$ be a fixed integer. If following conditions hold with correct constants (dependent on $d$),
	\[ 
		\sigma = \Omega(\log n) \qquad
		\tau = \Omega\brac{\sigma^{\frac{d}{d-1}}} \qquad
		c = O\brac{\tfrac{\log \sigma}{\log \log \sigma}}
	\]
	Then we have with high probability for all $v \in \overline{L}$, either 
	\[ \abs{\overline{I}_v} + \abs{\overline{O}_v} \ge t = \frac{n}{\tau},\]
	or else that 
	\[ \comb_c(\XXX_n \mid I_v, O_v) \le \tsp(\XXX_n) \]
	for large enough $n$.
\end{lemma}
\begin{proof}
	We will show that if $\abs{\overline{I}_v} + \abs{\overline{O}_b} \le t = \tfrac{n}{\tau}$, then $\comb_c(\XXX_n \mid I_v, O_v) \le \tsp(\XXX_n)$. 
	This proof has two components. First, we upper bound $\tspbar(\XXX_n \mid I_v, O_v)$ given that $\overline{I}_v, \overline{O}_v$ are small. More precisely, we will show that
	\begin{equation}
		\label{eq:conditional_tspbar}
		\tspbar(\XXX_n \mid I_v, O_v) \le \tspbar(\XXX_n) + O\brac{n^{\frac{d-1}{d}} \sigma^{\frac{d+1}{d}} \tau^{-1}}
	\end{equation}
	which follows from making local modifications to the optimal tour in $\overline{\Lambda}$.
	In the second part, we will bound the value of $\comb_c(\XXX_n \mid I_v, O_v)$ given that $\overline{I}_v, \overline{O}_v$ are small. In particular, we have:
	\begin{equation}
		\label{eq:conditional_comb}
		\comb_c(\XXX_n \mid I_v, O_v) \le \tspbar(\XXX_n \mid I_v, O_v) - O\brac{\brac{e^{-O(c \log c)} - \frac{1}{\tau} - \frac{1}{\sigma}} n^{\frac{d-1}{d}}}
	\end{equation}
	The proof of the second part it similar to that of \Cref{thm:comb_sep}.

	For the first part, notice that there are at most $t$ integers $j$ such that $\abs{I_v^j} + \abs{O_v^j} > 0$. We shall use the term \emph{restricted boxes} to denote all such boxes $B_j$.
	We construct a tour by modifying the optimal tour in $\overline{\Lambda}$, by replacing the portion of tour by any feasible tour in all the restricted boxes.
	Note that if there are no feasible tour in any of the boxes, then $\overline{\Lambda}_v = \emptyset$, and hence $v \notin \overline{L}$, which is a contradiction.
	Therefore, such a patching always exists. 

	In the restricted boxes, the total length of the tour can be the worst case length of the tour, which is $\beta_j s^{-1/d} \sqrt{d}$. Since with high probability, $\beta_j \le 2 \sigma$ for all $j$, we can conclude that
	\[ \tspbar(\XXX_n \mid I_v, O_v) \le \tspbar(\XXX_n) + 2 \sigma s^{-\frac{1}{d}} \frac{n}{\tau} = \tspbar(\XXX_n) + O\brac{n^{\frac{d-1}{d}} \sigma^{\frac{d+1}{d}} \tau^{-1}} \]

	On the other hand, following the proof of \Cref{l:smiley_face_lemma}, where we ensure that the smaller boxes of side-length $3D$ are contained in the boxes of side-length $s^{-1/d}$, we can find $e^{-O(c \log c)} n$ $(\varepsilon,D)$-copies of the gadget $\Pi^3_S(k)$, scaled by $n^{-1/d}$ (Note that $\varepsilon$ and $D$ also gets scaled by a factor $n^{-1/d}$). Here $\varepsilon = \Omega\brac{1/c}$ and $D = D_2$ is the absolute constant specified in \Cref{l:delta_single_visit}.
	
	Observe that \Cref{l:smiley_face_lemma} holds only when $\exp(O(c \log c)) = o(n)$, which is ensured since $c \log c = O(\log \sigma)$, and by choosing the constants correctly, we can ensure that $\exp(O(c \log c)) = O(\sigma^K) = o(n)$, where $K$ is an appropriately chosen constant.

	We look at the optimal TSP tour which has length $\tspbar(\XXX_n \mid I_v, O_v)$.
	We cannot directly use \Cref{l:parallel_exits,l:two_exits,l:pi_nice_corner,l:pi_s_nice_corner,l:delta_single_visit,l:nice_s_existance} on this tour to construct a solution that satisfies $\comb_c$, since the optimal tour in $\overline{\Lambda}$ might not by an optimal TSP tour.

	But, for any $(\varepsilon, D)$-copy of the gadget $S_1$ which is contained in box $B_j$, all the results will go through as long as we can perform the modification used in the proofs of \Cref{l:parallel_exits,l:two_exits,l:pi_nice_corner,l:pi_s_nice_corner,l:delta_single_visit,l:nice_s_existance} and get a tour that is contained in $\overline{\Lambda}_v$.
	These modifications can be made as long as any of the points in the $(\varepsilon, D)$-copy of the gadget or the points adjacent to these gadget are not contained in an edge in $\overline{I}_v$ or $\overline{O}_v$, and are not one of the special points $x_j^{\set{1,2,3,4}}$ used in definition of $\overline{\Lambda}$. Therefore, we can make these modifications on all but $O(s+t)$ gadgets!
	Therefore, we can construct a half-integral solution which satisfied $\comb_c$ constraints and respects the sets $I_v$ and $O_v$ of value at most
	\[ \tspbar(\XXX_n \mid I_v, O_v) - O\brac{\brac{e^{-O(c \log c)} - \frac{1}{\tau} - \frac{1}{\sigma}} n^{\frac{d-1}{d}}} \]
	In particular, this proves \Cref{eq:conditional_comb}.

	The condition $\tau = \Omega(\sigma^{d/(d-1)})$ along with \Cref{eq:conditional_tspbar,l:karp_dissection} implies that
	\[ \tspbar(\XXX_n \mid I_v, O_v) \le \tsp(\XXX_n) + O\brac{n^{\frac{d-1}{d}} \sigma^{-\frac{1}{d(d-1)}}} \]
	which again along with $\tau = \Omega(\sigma^{d/(d-1)})$ and $\sigma = \omega(1)$ gives us
	\[ \comb_c(\XXX_n \mid I_v, O_v) \le \tsp(\XXX_n) + n^{\frac{d-1}{d}} O\brac{\sigma^{-\frac{1}{d(d-1)}} - e^{-O(c \log c)}} \]
	We can now choose 
	\[ c = O\brac{\frac{\log \sigma}{d(d-1) \log \log \sigma}} = O\brac{\frac{\log \sigma}{\log \log \sigma}} \]
	to get that $\comb_c(\XXX_n \mid I_v, O_v) \le \tsp(\XXX_n)$ holds for large enough $n$.
\end{proof}

If $v \in \overline{L}$, then $\comb_c(\XXX_n \mid I_v, O_v) \ge \tsp(\XXX_n)$ and hence, by the result above, we must have that 
\[ \abs{\overline{I}_v} + \abs{\overline{O_v}} \ge \frac{n}{\tau} \]
Then by \Cref{eqn:bound_l_v_inc} and \Cref{eqn:bound_l_v_exc} gives
\[ \overline{\Lambda}_v \le \overline{\Lambda} e^{-\Omega\brac{\frac{n}{\sigma^3 \tau}}} \]
which implies that
\[ \abs{\overline{L}} \ge e^{\Omega\brac{\frac{n}{\sigma^3 \tau}}} \]

Observe that for a constant $c$, choosing $\sigma = K \log n$ and $\tau = \sigma^{d/(d-1)}$ gives us that 
\[ \abs{\overline{L}} \ge \exp\brac{\Omega\brac{\frac{n}{\log n^{4 + \frac{d}{d-1}}}}} = e^{\Omega\brac{\frac{n}{\log^5 n}}} \]
	which recovers \Cref{thm:bnb_restated}, that is
\tbnb*

Further, we get the exact same bound on the number of leaves when $c = O\brac{\tfrac{\log \log n}{\log \log \log n}}$. Similarly, for any $\varepsilon > 0$ we can set $\sigma = n^{\varepsilon}$ and $\tau = n^{\varepsilon d / (d-1))}$ to get that for any $c = O\brac{\tfrac{\varepsilon \log n}{\log \log n}}$ we have
\[ \abs{\overline{L}} \ge \exp\brac{\Omega\brac{\frac{n^{1 - \varepsilon}}{n^{\varepsilon \brac{4 + \frac{d}{d-1}}}}}} = e^{\Omega(n^{1 - 6 \varepsilon})} \]
which recovers \Cref{thm:bnb_restated_general}, that is
\tbnbg*

\newpage
\bibliographystyle{plain}
\bibliography{main.bib}

\newpage
\appendix
\section{Quantitative Bounds for Properties of Gadgets}
\label{s:distance_bounds}
This section will provide quantitative bounds to some properties of $(\varepsilon,D)$-copies of a gadget $T$. We will give bounds on $\varepsilon$ and $D$ in order to satisfy \Cref{obs:gadget_entry_bound}, and a slightly stronger version of it.
First, we set up some notation.
Given a sets $S \subseteq V \subseteq \R^2$ and an Hamiltonian path $P$ on $V$, we say that $S$ is connected to $V \setminus S$ through a pair of edges $e_1, e_2$ in $P$ if $e_1, e_2 \in \delta(S,V \setminus S)$, and $e_1$ and $e_2$ are connected in $P$ through a path completely contained in $S$.
\begin{lemma}
	\label{l:parallel_exits}
	Let $S$ be a gadget with diameter $d$, and let $P$ be an optimal Hamiltonian path through $V$.
	Given $\varepsilon > 0$ and $\theta > 0$, there is $D \ge D(\varepsilon,\theta,d)$ such that if $S_1$ is any $(\varepsilon,D)$-copy of the gadget $S$ such that there are two or more pairs of edges joining $S_1$ to $V \setminus S_{\varepsilon,D}$ in $P$ then the angle between any connecting pair of edges is at least $\pi - \theta$.
	In particular,
	\begin{equation}
		\label{eq:parallel_exits_bound}
		D(\varepsilon,\theta,d) =  \frac{6d+12\varepsilon}{1 - \cos \theta}
	\end{equation}
	suffices.
\end{lemma}
\begin{proof}
	Suppose $e_1, e_2$ is a pair of edges connecting $S_1$ to $V \setminus S_1$. Let $e_i = \set{p_i,x_i}$ where $x_i \in S_1$, $p_i \notin S_1$ for $i = 1,2$. First, we make a precise definition of the angle between these two edges using the cosine formula.
	
	\begin{definition}
		\label{d:angle}
		The angle between $\overrightarrow{x_1 p_1}$ and $\overrightarrow{x_2 p_2}$ denoted by $\measuredangle(\overrightarrow{x_1 p_1}, \overrightarrow{x_2 p_2})$ is the angle $\phi \in [0,\pi]$ such that
		\[ \cos(\phi) = \frac{\inprod{x_1 p_1}{x_2 p_2}}{\norm{x_1p_1} \cdot \norm{x_2p_2}} \]
	\end{definition}

	Let $\phi = \measuredangle(\overrightarrow{x_1 p_1}, \overrightarrow{x_2 p_2})$. Let $f_1, f_2$ be any other pair of edges connecting $S_1$ to $V \setminus S_1$. Let $f_i = \set{q_i, y_i}$ where $y_i \in S_1, q_i \notin S_1$ for $i = 1,2$. Since $P$ is optimal Hamiltonian path, \emph{short-cutting} $p_1,p_2$ must give a longer path. To be precise, the path $Q$ obtained by deleting edges $e_1,e_2$, $y_1 z$ where $z \neq q_1$, and adding edges $p_1 p_2$, $y_1 x_1$, $x_2 z$, is longer than the path $P$.
	In particular, we must have
	\begin{equation}
		\label{eq:parallel_exits_1}
	  \ell (p_1 p_2) + 2d + 4\varepsilon \ge \ell (p_1 x_1) + \ell (p_2 x_2)
	\end{equation}
	Let $p_2'$ be a point such that $x_1y_1p_2p_2'$ is a parallelogram. Therefore, $\ell (p_1p_2) \le \ell (p_1p_2') + d + 2\varepsilon$. Hence, it must hold that
	\begin{equation}
		\label{eq:parallel_exits_2}
		\ell (p_1 p_2') + 3d + 6\varepsilon \ge \ell (p_1 x_1) + \ell (p_2 x_2)
	\end{equation}
	Let $a = \ell (p_1 x_1), b = \ell (p_2 x_2), c = \ell (p_1 p_2')$. Then by definition of $\phi$,
	\[ c^2 = a^2 + b^2 - 2ab \cos \phi \]
	Using this, we get
	\begin{align*}
	  \ell (p_1 x_1) + \ell (p_2 x_2) - \ell (p_1 p_2')
		& = \frac{(a+b)^2 - c^2}{a + b + c} \\
		& \ge \frac{(a+b)^2 - c^2}{2(a+b)} \tag*{Since $a+b \le c$}\\
		& = \frac{2ab(1 + \cos \phi)}{2(a+b)} = \frac{ab(1-\cos \phi)}{a+b}
	\end{align*}
	Since $S_1$ is an $(\varepsilon,D)$ copy of $S$, $a,b \ge D$. Under this condition $\tfrac{ab}{a+b}$ is minimized at $a = b = D$, implying that $\ell (p_1 x_1) + \ell (p_2 x_2) - \ell (p_1 p_2') \ge \tfrac{D(1 + \cos \phi)}{2}$. Hence, for \Cref{eq:parallel_exits_2} to hold, we must have
	\[ 3d+6\varepsilon \ge \frac{D(1 + \cos \phi)}{2} \]
	In particular, if 
	\[ D \ge \frac{6d+12\varepsilon}{1 - \cos \theta} \]
	then $1 + \cos \phi \le 1 - \cos \theta \implies \phi \ge \pi - \theta $, which completes the proof giving us the bound in \Cref{eq:parallel_exits_bound}.
\end{proof}

\begin{lemma}
	\label{l:two_exits}
	Let $S$ be a gadget with diameter $d$, and let $P$ be an optimal Hamiltonian path through $V$.
	Given $\varepsilon > 0$ and $\tfrac{\pi}{4} \ge \theta > 0$, there is $D \ge D(\varepsilon,\theta,d)$ such that if $S_1$ is any $(\varepsilon,D)$-copy of the gadget $S$ then there are at most $2$ pairs of edges joining $S_1$ to $V \setminus S_1$. Further, all the four edges joining $S_1$ to $V \setminus S_1$ make an acute angle of at most $2 \theta$ with each other.
	In particular, 
	\begin{equation}
		\label{eq:two_exits_bound}
		D(\varepsilon,\theta,d) =  \frac{6d+12\varepsilon}{1 - \cos \theta}
	\end{equation}
	suffices.
\end{lemma}
\begin{proof}
	Let $e_1,e_2$ be a pair of edges joining $S_1$ to $V \setminus S_1$ such that $e_i = \set{x_i, p_i}$ where $x_i \in S_1, p_i \notin S_1$ for $i = 1,2$.
	Let $f_1, f_2$ be a pair of edges joining $S_1$ to $V \setminus S_1$ such that $f_i = \set{y_i,q_i}$ where $y_i \in S_1, q_i \notin S_1$ for $i = 1,2$. Further, let $p_2$ and $q_1$ be through portion of $P$ that does not contain $x_2$.

	Since $P$ is an optimal Hamiltonian path, the Hamiltonian path $Q$ obtained by deleting edges $p_1 x_1$, $q_1 y_1$ and adding edges $x_1 y_1$, $p_1 q_1$, must by as long. Therefore, we must have
	\[ d \ge \ell (p_1 x_1) + \ell (q_1 y_1) - \ell (x_1 y_1) \]
	By the computations in \Cref{l:parallel_exits}, for $D \ge \tfrac{6d + 12\varepsilon}{1 - \cos \theta}$, this hold only if $\measuredangle(\overrightarrow{x_1 p_1}, \overrightarrow{y_1 q_1}) \ge \pi - \theta$. This observation combined with \Cref{l:parallel_exits} implies that all four edges $e_1,e_2,f_1,f_2$ make an acute angle of at most $2\theta$ with each other (This holds even if they are not coplanar!).

	Now, assume that there is another pair of edges $g_1, g_2$ joining $S_1$ to $V \setminus S_1$, such that $g_i = \set{z_i,r_i}$ where $z_i \in S_1, r_i \notin S_1$ for $i = 1,2$ and $q_2$ and $r_1$ are connected through portion of $P$ that does not contain $y_2$. Then we have
	\begin{align*}
	\measuredangle(\overrightarrow{x_1 p_1}, \overrightarrow{y_1 q_1}) & \ge  \pi - \theta \\
	\measuredangle(\overrightarrow{y_1 q_1}, \overrightarrow{r_1 z_1}) & \ge  \pi - \theta \\
	\measuredangle(\overrightarrow{x_1 p_1}, \overrightarrow{r_1 z_1}) & \ge  \pi - \theta
	\end{align*}
	This leads to contradiction, since first two equations imply $\overrightarrow{x_1,p_1}$ and $\overrightarrow{r_1 z_1}$ are on the same side of hyperplane $\inprod{v}{q_1 - y_1} = 0$. But, the third equation implies otherwise!
\end{proof}

\section{Properties of Hamiltonian Paths in the Gadgets}
\label{s:gadget_proofs}
In this section, we will provide proofs of various geometrical lemma regarding properties of the gadgets in this section. These include proofs of \Cref{l:pi_nice_corner,l:pi_s_nice_corner,l:delta_single_visit}.

\subsection{Proof of \texorpdfstring{\Cref{l:pi_nice_corner}}{lemma \ref{l:pi_nice_corner}}}
Let us begin by recall definition of $\Pi(t,h,w)$ and $\Pi_S = \Pi(S,t,h,w)$ (\Cref{d:pi_t_h_w,d:pi_s}):
\dPiTHW*
\dPiS*
\noindent Now we are ready to provide proofs of lemmas in \Cref{ss:expanding_the_gadget}.
\label{ss:proof_pi_nice_corner}
\lPiNiceCorner*
\begin{proof}
	We begin with a few observations:
	\begin{observation}
		\label{obs:pi_vertex_order}
		If $P' = a v_{i_1} \ldots v_{i_k} \pi_i$ is a contiguous segment in $P$, then either $i_1 < \ldots < i_k$ or $i_k < \ldots < i_1$.
	\end{observation}
	\noindent Suppose not. Let $j_1, \ldots, j_k$ be a sorting of $i_1, \ldots, i_k$ in increasing order. Then $j_1$ and $j_k$ appear somewhere in $P'$. Suppose $j_1$ appears before $j_k$. For notational convenience, let $\ell (a_1 \ldots a_j)$ denote the length of the path $a_1, \ldots a_j$.
	\begin{align*}
		\ell (a v_{i_1} \ldots v_{i_k} \pi_i)
		& \ge \ell (av_{i_1}) + \ell (v_{i_1} v_{j_1}) + \ell (v_{j_1} v_{j_k}) + \ell (v_{j_k} v_{i_k}) + \ell (v_{i_k} \pi_i) \\
		& \ge \ell (av_{j_1}) + \ell (v_{j_1} v_{j_k}) + \ell (v_{j_k}\pi_i) \tag*{Triangle Inequality}\\
		& \ge \ell (av_{j_1} \ldots v_{j_k}\pi_i)
	\end{align*}
	Similarly, in the case when $j_k$ appears before $j_1$, we get
	\[ \ell (a v_{i_1} \ldots v_{i_k} \pi_i) \ge \ell (a v_{j_k} \ldots v_{j_1} \pi_i) \]
	\begin{observation}
		If $P' = av_{i_1} \ldots v_{i_k}b$, then we can assume that $i_1 \ldots i_k$ is a continuous subset of $[t]$.
	\end{observation}
	\noindent First, we can by \Cref{obs:pi_vertex_order} assume $i_1, \ldots, i_k$ are sorted either in increasing order or decreasing order. Without loss of generality, let $i_1 < i_k$. Further, let $p$ be an index such that $i_1 < p < i_k$ that is not contained in the set $\set{i_1, \ldots, i_k}$. Then we can insert $v_p$ into $v_{i_1} \ldots v_{i_k}$ without changing the total length of the portion $P'$. On the other hand, shortcut through $v_p$ in $P$ whenever $v_p$ was present may decrease the total length. Thus, this replacement can only get us a shorter path.

	\noindent Using the two observations, we can assume that the shortest Hamiltonian path $P$ looks like this:
	$p \overline{v_{i_1} v_{j_1}} c_1 \overline{v_{i_2} v_{j_2}} c_2 \ldots c_4 \overline{v_{i_5} v_{j_5}} q$
	Where by $\overline{v_{i_1} v_{j_1}}$ we mean the path containing all the vertices between $v_{i_1}$ and $v_{j_1}$.
	Let $\mathcal{C} = \set{\pi_1,\pi_2,\pi_3,\pi_4}$ denote the set of four corners.
	\begin{observation}
		\label{obs:corner_improve}
		Let $p$ such that $\dist(p, \Pi(t,h,w)) \ge D$, and let $v_i, v_j$ be any points in $\set{v_1, \ldots, v_t}$.
		Then if $D \ge \tfrac{h^2 + w^2}{4w}$ and $h \ge 2w$ then
		\begin{equation}
			\label{eq:corner_improve}
			\ell (p v_i) + \ell (v_j c) \ge \dist(p, \mathcal{C}) + \frac{w}{4}
		\end{equation}
		for $c \in \set{\pi_1,\pi_2,\pi_3,\pi_4}$.
	\end{observation}
	Suppose $p = (x_1,y_1) \in \R^2$. We will prove the result by working on different cases based on $(x_1,y_1)$.
	\begin{case}
		\label{case1:corner_improve}
		$y_1 \ge h$: Without loss of generality, assume that $x_1 \ge 0$. Then $\ell (p v_i) \ge \ell (p v_t)$ and $\ell (v_j c) \ge \frac{w}{2} = \ell (v_t \pi_3)$. Therefore, by triangle inequality,
		\[ \ell (p v_i) + \ell (v_j c) \ge \ell (p v_t) + \ell (v_t \pi_3) = \frac{w}{2} + \ell (p v_t) \]
		If $\ell (p v_t) \ge \ell (p \pi_3)$, we get the result in this case.
		Therefore, we can assume that $x \le \frac{w}{4}$. Since $\ell (p v_t) \ge (y_1 - h)$, we it suffices to show that
		\[ \brac{\bar{y}_1 + \frac{w}{2} - \frac{w}{4}}^2 \ge \dist(p, \mathcal{C})^2 = \bar{y}_1^2 + \brac{x - \frac{w}{2}}^2 \]
		where $\bar{y}_1 = y_1 - h$. Since $0 \le x_1 \le \tfrac{w}{4}$, it suffices to show that 
		\[ \brac{\bar{y}_1 + \frac{w}{4}}^2 \ge \bar{y}_1^2 + \frac{w^2}{4} \]
		This is satisfied when $y' \ge \tfrac{3w}{8}$. Since $\dist(p, \Pi(t,h,w)) \ge \bar{y}_1$, this holds when $D \ge \tfrac{h^2 + w^2}{4w}$ and $h \ge 2w$. 
	\end{case}
	\begin{case}
		\label{case2:corner_improve}
		$y_1 \le 0$: This case holds due to computations similar to \Cref{case1:corner_improve}.
	\end{case}
	\begin{case}
		\label{case3:corner_improve}
		$0 \le y_1 \le h$ and $x_1 > 0$: In this case $\ell (x v_i) \ge x_1$ and $\ell (v_j c) \ge \frac{w}{2}$, therefore, \Cref{eq:corner_improve} holds if and only if
		\[
			\brac{x_1 + \frac{w}{4}}^2 \ge \brac{x_1 - \frac{w}{2}}^2 + y_1^2
			\qquad \text{or} \qquad
			\brac{x_1 + \frac{w}{4}}^2 \ge \brac{x_1 - \frac{w}{2}}^2 + \brac{y_1 - h}^2
		\]
		We will look at the region where a stronger condition holds, namely
		\[
			x_1^2 \ge \brac{x_1 - \frac{w}{2}}^2 + y_1^2
			\qquad \text{or} \qquad
			x_1^2 \ge \brac{x_1 - \frac{w}{2}}^2 + \brac{y_1 - h}^2
		\]
		These constraints define region bounded by parabolas, and point of intersection of these two parabolas is the point furthest away from $\Pi(t,h,w)$ where both the conditions fail.
		The point of intersection of the parabolas is given by $p = \brac{\tfrac{h^2 + w^2}{4w},\tfrac{h}{2}}$. 
		Therefore, \Cref{eq:corner_improve} holds for all points $p$ satisfying $x_1 \ge \tfrac{h^2 + w^2}{4w}$.
		Since all points outside both the parabolas satisfy $x_1 \ge \dist(p, \mathcal{C})$, result holds for $D = \tfrac{h^2 + w^2}{4w}$, since 
	\end{case}
	\begin{case}
		\label{case4:corner_improve}
		$0 \le y_1 \le h$ and $x_1 < 0$: Following the same computations as in \Cref{case3:corner_improve}, we get the exact same condition on $D$.
	\end{case}
	Now we are ready to prove structure of $P$, but first we need one definition.
	\begin{definition}
		\label{d:distance_without_center}
		Consider any Hamiltonian path $P$ that looks like
		$p \overline{v_{i_1} v_{j_1}} c_1 \overline{v_{i_2} v_{j_2}} c_2 \ldots c_4 \overline{v_{i_5} v_{j_5}} q$.
		For a subpath $p' \overline{v_i v_j} q'$ of $P$, where $p',q' \in \set{p,q,c_1,c_2,c_3,c_4}$, we define $d(p'q')$ as follows:
		\begin{itemize}[nosep]
			\item $d(p'q') = \ell (p' v_i) + \ell (q' v_j)$ if $\overline{v_i v_j} \neq \emptyset$
			\item $d(p'q') = \ell (p'q')$ if $\overline{v_i v_j} = \emptyset$
		\end{itemize}
	\end{definition}
	\Cref{obs:corner_improve} implies that $d(p,c_1) \ge \dist(p,\mathcal{C})$. Further, for $1 \le a \le 3$, we have $d(c_\alpha, c_{\alpha+1}) \ge \min(h,w) = w$, since if $\overline{v_{i_{\alpha+1}} v_{j_{\alpha+1}}} \neq \emptyset$, $\ell (c_\alpha v_{i_{\alpha+1}}) + \ell (v_{j_{\alpha+1} c_{\alpha+1}}) \ge \frac{w}{2} + \frac{w}{2} = w$.
	There for we have the lower bound on length of any optimal Hamiltonian path $P$ from $p$ to $q$:
	\[ d(p,c_1) + d(c_1,c_2) + d(c_2,c_3) + d(c_3,c_4) + d(c_4,q) + \sum_{i=1}^{5} l(v_{i_1} v_{j_1}) \ge \dist(p,\mathcal{C}) + \dist(q,\mathcal{C}) + 3w + h\brac{1 - \frac{4}{t}} \]
	Note that since $p,q$ are on different sides of line $y = \frac{h}{2}$, the nearest corners from $p,q$ respectively are different and are not on the same short side of the gadget. Therefore, we can construct a Hamiltonian path $Q$ such that 
	\[ \ell (Q) \le \dist(p,\mathcal{C}) + \dist(q,\mathcal{C}) + 3w + h \]
	In the path $P$, if the path $c_1c_2c_3c_4$ contains two longer sides of the gadget, then we have 
	\[ \ell (P) \ge \dist(p,\mathcal{C}) + \dist(q,\mathcal{C}) + w + 2h \]
	which is longer that $Q$ if $h \ge 2w$.
	\Cref{obs:corner_improve} further implies that if $\overline{v_{i_1} v_{j_1}} \neq \emptyset$, then 
	\[ \ell (P) \ge \dist(p,\mathcal{C}) + \dist(q,\mathcal{C}) + 3w + \frac{w}{4} + h - \frac{4h}{t} \] 
	Therefore, when $\frac{w}{4} \ge \frac{4h}{t}$ or equivalently $t \ge \frac{16h}{w}$, $\overline{v_{i_1} v_{j_1}} = \emptyset$ and $\overline{v_{i_5} v_{j_5}} = \emptyset$. Thus, the shortest Hamiltonian path $P$, is determined by choice of $\overline{v_{i_\alpha} v_{j_\alpha}}$ for $\alpha = 2,3,4$.
	Suppose without loss of generality that $c_1 c_2$ is the shorter side of the gadget given by $y=0$. Then the values of $i_\alpha, j_\alpha$ that minimize $d(c_1 c_2) + d(c_2 c_3) + d(c_3 c_4)$ are given by ${i_2} = {j_2} = 0, i_3 = 1, j_3 = t-1, i_4 = j_4 = t$.
	This completely describes the shortest Hamiltonian path $P$, and both points $c_2, c_3$ satisfy the condition in the lemma, completing the proof.
\end{proof}

\subsection{Proof of \texorpdfstring{\Cref{l:pi_s_nice_corner}}{lemma \ref{l:pi_s_nice_corner}}}
\label{ss:proof_pi_s_nice_corner}
\lPiSNiceCorner*
\begin{proof}
	Let $\pi_i$ denote the center of $S_i$. Let $\mathcal{S} = \bigcup_{i=1}^4 S_i$ and $\mathcal{C} = \set{\pi_1, \ldots, \pi_4}$.
	Note that \Cref{obs:corner_improve} holds with when $D \ge \tfrac{h^2 + w^2}{4w}$ with 
	\[ \ell (p v_i) + \dist(v_j S) \ge \dist(p S) + \frac{w}{4} - 8 \]
	Since $\ell (p v_i) + \ell (v_j c) \ge \ell (p c) + \tfrac{w}{4}$ for center $c$ of the gadget $S$ and $\dist{v_j S} \ge \ell (v_j c) - 4$ and $\ell (p c) \ge \dist(p S) - 4$.
	Further, we can extend the path that we obtain in the proof of the previous lemma by including an Hamiltonian path through $S_i$ when the path is supposed to visit $\pi_i$ to get a Hamiltonian path $P_1$ from $p$ to $q$ with length at most
	\begin{equation}
		\label{eq:p_1_upper_bound}
		\ell (P_1) \le \dist(p, \mathcal{S}) + \dist(q, \mathcal{S}) + 3(w-8) + h + 4(10 \pi + 8) + 16
	\end{equation}
	since length of tour in each gadget is $10\pi + 8$, actual distance between two closest gadgets is $w - 8$, and since we must enter and exit next in adjacent vertices to extend the tours as defined in \Cref{fig:single_entry}, we pay an additional factor of $8$.
	Now, we extend \Cref{d:distance_without_center} to sets:
	\begin{definition}
		\label{d:distance_without_center_set}
		Given a Hamiltonian path $P$ in $\set{p,q} \cup \Pi(S(k),t,h,w)$ from $p$ to $q$, which can be represented as
		$pv_{i_1} v_{j_1}T_1 \ldots v_{i_u} v_{j_u} T_u v_{i_{u+1}} v_{j_{u+1}} q$
		where for each $i$, $T_i$ is a path such that $T_i \subseteq S_j$ for some $j \in \set{1,\ldots,4}$.
		For any two sets $R_1, R_2 \in \set{\set{p},\set{q}, T_1, \ldots, T_u}$, such that there is a subpath
		$p'\overline{v_i v_j} q'$ in $P$, we define $d(R_1,R_2)$ as follows:
		\begin{itemize}[nosep]
			\item $d(R_1 R_2) = \dist(R_1 v_i) + \dist(R_2 v_j)$ if $\overline{v_i v_j} \neq \emptyset$
			\item $d(R_1 R_2) = \dist(R_1 R_2)$ if $\overline{v_i v_j} = \emptyset$
		\end{itemize}
	\end{definition}
	\begin{observation}
		\label{obs:p_visits_once}
		There is an absolute constant $C$ such that when $k \ge C$, then $P$ visits each $S_i$ exactly once.
	\end{observation}
	We can write $P$ as 
	$p\overline{v_{i_1} v_{j_1}}T_1 \ldots \overline{v_{i_u} v_{j_u}} T_u \overline{v_{i_{u+1}} v_{j_{u+1}}} q$,
	where $T_i$ is a path such that $T_i \subseteq S_j$ for some $j \in \set{1,\ldots,4}$.
	Then we have $d(T_\alpha, T_{\alpha+1}) \ge w - 8$ and $d(\set{p},T_1) \ge \dist(p,\mathcal{S})$ and $(\set{q},T_u) \ge \dist(q, \mathcal{S})$. Note that each point in $S_i$ must still be connected to some vertex, and sum of the distances between each vertex and it's nearest neighbor is $40\pi$. This gives the lower bound:
	\begin{equation}
		\label{eq:p_lower_bound}
		\ell (P) \ge \dist(x, \mathcal{C}) + \dist(y, \mathcal{C}) + (u-1)(w-8) + h + 4 \brac{10 \pi + 8 - O\brac{\frac{1}{k}}} - O\brac{\frac{u}{k}}
	\end{equation}
	The additive correction $O\brac{\tfrac{u}{k}}$ is to account for double counting. All the vertices in $\mathcal{S}$ that are connected to something outside are counted twice, once in $40\pi$ and once in $(u-1)(w-8)$. We must subtract their contribution in the $40\pi$ term, which is at most $\tfrac{4\pi}{k}$ for each vertex. Since number of these connecting vertices is at most $2u$, we get the additive correction factor, with $8\pi$ being the constant hidden in $O$-notation.
	Observe that Since $P$ is a shortest Hamiltonian path, it is shorter than $P_1$, and hence we must have
	\[ \brac{u-4}\brac{w-8 - \frac{8\pi}{k}} - \frac{32 \pi}{k} - 16 \le 0 \]
	It follows that $u \le 4$ if $w \ge 100$ for $k \ge 16 \pi$. This finishes the proof of \Cref{obs:p_visits_once}.

	Therefore, $P$  looks like
	$p\overline{v_{i_1} v_{j_1}}T_1 \ldots \overline{v_{i_4} v_{j_4}} T_4 \overline{v_{i_{5}} v_{j_{5}}} q$.
	If $\overline{v_{i_1} v_{j_1}} \neq \emptyset$ then \Cref{obs:corner_improve} gives a better lower bound on $\ell (P)$. In particular, it increases the lower bound in \Cref{eq:p_lower_bound} by $\tfrac{w-8}{4}$. Comparing this lower bound on $\ell (P)$ with upper bound on $\ell (P_1)$  given in \Cref{eq:p_1_upper_bound}, following must hold
	\[ \frac{w-8}{4} - 16 - O\brac{\frac{1}{k}} - \frac{4h}{t} \le 0  \]
	This fails to hold when $w \ge 100$ and $t \ge 2h$ for large enough $k$. Hence, we can conclude that $\overline{v_{i_1} v_{j_1}} = \overline{v_{i_5} v_{j_5}} = \emptyset$.
	Hence, if $\tfrac{h}{t} \approx \frac{4\pi}{k}$, then we can change $P$ to $Q$ by replacing tour inside $T_2$ and $T_3$ by the Hamiltonian path described in \Cref{fig:single_entry}, and connecting it to it's nearest neighbors among $v_1, \ldots, v_t$, which are either $\set{v_1,v_2}$ or $\set{v_{t-1},v_t}$ by choice of $t$. Note that this replacement strictly reduces the total cost outside the gadget, and is optimal inside the gadget up to an additive factor of $O(1/k)$. Therefore, we get the path $Q$ such that
	\[ \ell (Q) \le \ell (P) + O\brac{\frac{1}{k}} \]
\end{proof}

\subsection{Proof of \texorpdfstring{\Cref{l:delta_single_visit}}{lemma \ref{l:delta_single_visit}}}
\label{ss:proof_delta_single_visit}
\lDelta*
\begin{proof}
	Since we choose $\Pi_S(k) = \Pi(S(k), \tfrac{200k}{4\pi}, 200, 100)$, the diameter of $\Pi_S(k)$ is at most $300$.
	Let
	\[ D_1 = \frac{2000}{1 - \cos \frac{\pi}{10}} \]
	be chosen to satisfy conditions of \Cref{l:parallel_exits,l:two_exits} for $\Pi_S(k)$ and $\theta = \tfrac{\pi}{10}$. Then $\Delta(D_1, \Pi_S(k))$ has diameter at most $\tfrac{5000}{1 - \cos(\pi/10)}$.
	Let
	\[ D_2 = \frac{30000}{\brac{1 - \cos \frac{\pi}{10}}^2} \]
	be chosen to satisfy conditions of \Cref{l:parallel_exits,l:two_exits}  for $\Delta(D_1, \Pi_S(k))$ and $\theta = \tfrac{\pi}{10}$.
	It follows that $\Delta_1$ and $\Pi_i$ for $i = 1,2,3$ can be visited by $P$ at most twice, and if they are visited by $P$ exactly twice, then all the four edges exiting the corresponding set are nearly parallel. We will say that $P$ connects two sets $X,Y \subseteq V$ if and only if $P$ contains an edge going from $X$ to $Y$. Now, we do cases based on how many times these sets are visited.

	\case Suppose that there is $\Pi_i$ such that $P$ visits $\Pi_i$ twice. Without loss of generality, we will assume that $P$ visits $\Pi_1$ twice. Let $e_1,e_2$ and $f_1,f_2$ be two pairs of edges connecting $\Pi_1$ to $V \setminus \Pi_1$. If $g_1$ connects $\Pi_1$ to $\Pi_2$, and $g_2$ connects $\Pi_1$ to $\Pi_3$, where $g_1, g_2 \in \set{e_1,e_2,f_1,f_2}$, then $g_1$ and $g_2$ have an acute angle of at most $\tfrac{\pi}{3}$ between them. Since $\tfrac{\pi}{3} \ge \tfrac{\pi}{5}$, this contradicts \Cref{l:two_exits}. Hence, $P$ connects $\Pi_1$ to exactly one of $\Pi_2, \Pi_3$. 

	\subcase If $\Pi_1$ is connected to neither $\Pi_2, \Pi_3$, then $P$ visits $\Delta_1$ at least $3$ times, twice in $\Pi_1$, and once in $\Pi_2 \cup \Pi_3$, which is a contradiction to \Cref{l:two_exits}. Without loss of generality, let $\Pi_1$ be connected to $\Pi_2$. Note that if $\Pi_2$ is not connected to $\Pi_3$, then $P$ visits $\Delta_1$ at least thrice, twice in $\Pi_1 \cup \Pi_2$, and at least once in $\Pi_3$.

	\subcase If $P$ visits $\Pi_2$ twice, then $\Pi_2$ cannot be connected to $\Pi_3$, since it is already connected to $\Pi_1$, which is a contradiction.

	\subcase If $P$ visits $\Pi_2$ exactly once, then $P$ must connect $\Pi_2$ to both $\Pi_1, \Pi_3$, and since $\Pi_1$ and $\Pi_3$ are on opposite sides of $\overleftrightarrow{OC_2}$, $i = 2$ satisfies all the conditions of the lemma. 
	
	\case Suppose that each of $\Pi_1, \Pi_2, \Pi_3$ is visited exactly once. Now, we have two cases based on how many times $\Delta_1$ is visited.

	\subcase If $\Delta_1$ is visited exactly once, then $P$ must visit $\Pi_1, \Pi_2, \Pi_3$ in some order, covering the whole set. Suppose this order is $\Pi_{j_1} \Pi_{j_2} \Pi_{j_3}$. Then $i = j_2$ satisfies all the conditions of lemma, since $\Pi_{j_1}$ and $\Pi_{j_3}$ are on opposite side of $\overleftrightarrow{OC_{j_2}}$.

	\subcase If $\Delta_1$ is visited twice, then let $P$ intersect $\Delta_1$ in two contiguous subpaths, say $Q_1, Q_2$. Without loss of generality, suppose that $\Pi_1 \subseteq Q_1$ and $\Pi_2, \Pi_3 \subseteq Q_2$.
	Let $e_1,e_2$ be pair of edges that connects $\Pi_1$ to $V \setminus \Pi_1$. Let $e_i = \set{x_i,p_i}$ where $p_i \notin \Pi_1$, and $x_i \in \Pi_1$. By \Cref{l:parallel_exits}, $\measuredangle(\overrightarrow{x_1p_1},\overrightarrow{x_2p_2}) \in \pi \pm \tfrac{\pi}{10}$.
	If possible, let $p_1, p_2$ be on the same side of $\overleftrightarrow{OC_1}$. Further, without loss of generality, let $\measuredangle(\overrightarrow{C_1O}, \overrightarrow{C_1p_1}), \measuredangle(\overrightarrow{C_1O},\overrightarrow{C_1p_2}) \in [0,\pi]$.
	Let $\theta_1 = \measuredangle(\overrightarrow{C_1O}, \overrightarrow{x_1 p_1})$ and $\theta_2 = \measuredangle(\overrightarrow{C_1O}, \overrightarrow{x_2 p_2})$. Since $p_1,p_2$ are on the same side of $\overleftrightarrow{OC_1}$, we must have
	\[ d + D_2 \sin \theta_1 \ge 0 \qquad d + D_2 \sin \theta_2 \ge 0 \]
	This implies that $\sin \theta_i \ge - \tfrac{d}{D_2}$. Since $\abs{\theta_1 - \theta_2} \in \pi \pm \tfrac{\pi}{10}$, it implies that 
	\[ \theta_i \in \sbrac{-\frac{\pi}{9}, \frac{\pi}{9}} \cup \sbrac{\pi - \frac{\pi}{9}, \pi + \frac{\pi}{9}} \]
	In fact, each of the two intervals contains exactly one $\theta_i$. Suppose $\theta_1 \in \sbrac{-\tfrac{\pi}{9}, \tfrac{\pi}{9}}$.
	Observe that for any $y_2 \in \Pi_2$, $\measuredangle(\overrightarrow{C_1 0}, \overrightarrow{x_1 y_2}) \le - \tfrac{\pi}{7}$ and for any $y_3 \in \Pi_3$, $\measuredangle(\overrightarrow{C_1 0}, \overrightarrow{x_1 y_3}) \ge \tfrac{\pi}{7}$. It follows that for any $y_2 \in \Pi_2$ and $y_3 \in \Pi_3$, $p_1$ is contained in $\angle y_2 x_1 y_3$. Since $\ell (x_1 y_2), \ell (x_1 y_3) \le D_1 + 4d \le D_2 \le \ell (x_1 p_2)$, the edge $e_1$ must intersect edge $y_2y_3$.
	Since $Q_2$ connects $\Pi_2,\Pi_3$, this implies that $e_1$ intersects and edge in $Q_2$, implying that $P$ is not planar!
	But since $P$ is the optimal Hamiltonian path, it must by planar, contradiction!
	
	This covers all the cases, completing the proof of lemma.
\end{proof}

\subsection{Proof of \texorpdfstring{\Cref{l:s_k_bound}}{lemma \ref{l:s_k_bound}}}
\label{ss:proof_s_k_bound}
Here we provide some more details for the proof of \Cref{l:s_k_bound} for sake of completeness.
\begin{lemma}
	\label{l:s_lower_bound}
	Consider the gadget $S = S(k)$ defined in \Cref{d:s_k} for large enough $k$. Let $p,q \in S$ be two points on the outer circle. Then the shortest Hamiltonian path from $p$ to $q$ completely covering $S$ has length at least $10 \pi + 8 - \frac{12\pi}{k}$.
\end{lemma}
\begin{proof}
	For this proof, we will approximate smaller segments along the circles by the arcs, the difference between them is $O(k^{-3})$, and since there are $O(k)$ of them, all the computations holds up to $O(k^2)$ error.

	Let $O_1$ denote the set of point on the inner circle of $S$ and let $O_2$ denote the set of points on the inner circle. Let $G = \set{g_1 = (-2,0), g_2 = (2,0)}$ be the set of gap vertices.
	Let $P$ be the shortest Hamiltonian path from $p$ to $q$ in $S$. To each vertex in $S$, we associate the length of the edge leaving that vertex in $P$ as the cost. Cost of each vertex in $O_1$ is at least $\tfrac{2\pi}{k}$ and cost of each vertex in $O_2$ is at least $\tfrac{4\pi}{k}$.
	Consider the path $P_1$ obtained by deleting $G$ from $P$. Then the path $P$ must leave and enter $O_2$ at least once, and the number of edges in $P$ that contain exactly one vertex in $O_2$ is even. Let $2t$ denote number of such edges. Thus, ever such edge costs at least $3 - \tfrac{4\pi}{k}$ additional length to the path $P_1$. This gives us the lower bound: 
	\[ \ell (P) \ge \ell (P_1) \ge 10 \pi + 2t\brac{3 - \frac{4\pi}{k}} \]
	For $k \ge \tfrac{4\pi}{3}$, this is an increasing function in $t$. Further, for $k \ge 4\pi$, value of this function at $t = 2$ is at least $10\pi + 8$. Therefore all the paths with $t \ge 2$ satisfy the required length condition.

	Suppose that $t = 1$, but the original path $P$ leaves $O_2$ more than once. Then, there must be a gap vertex that has both of it's neighbors in $O_2$. This implies $\ell (P) \ge \ell (P_1) + 4 - \tfrac{4\pi}{k}$. Since $t = 1$, we have $\ell (P_1) \ge 10 \pi + 6 - \tfrac{8\pi}{k}$, we get the bound
	\[ \ell (P) \ge 10\pi + 8 - \frac{12\pi}{k} \]
	which satisfies the requirement of the theorem.
	Similarly, if there is a vertex $g \in G$ such that both neighbors of $G$ lie in $O_1$, then this implies $\ell (P) \ge \ell (P_1) + 2 - \frac{4\pi}{k}$. This leads to exactly the same length bound as above.

	Hence, we are left with the case with path $P$ leaves and enters $O_2$ exactly once and both $g_1$ and $g_2$ have exactly one neighbor in $O_1$ and one in $O_2$.
	Suppose $p_1$ and $q_1$ are neighbors of $g_1$ and $g_2$ respectively in $O_1$. We claim that any Hamiltonian path $Q$ from $p_1$ to $q_1$ in $O_1$ must have length at least $\dist(p_1,q_1) + 2\pi - \frac{4\pi}{k}$. 

	Note that line $\overleftrightarrow{p_1 q_1}$ divides $O_1$ in two parts, say $H_1$ and $H_2$. For sake of notational convenience, we will include $p_1,q_1$ in both $H_1$ and $H_2$. Let $Q$ be denoted by $p_1 = v_0, \ldots, v_t = q_1$.
	For each $i$, define $\alpha_i$ to be the point in $H_1$ that is furthest away from $p_1$ and $\beta_i$ to be the point in $H_2$ that is furthest away from $p_1$. We claim that following holds for each $i$:
	\begin{enumerate}[nosep]
	  \item $v_i$ either equals $\alpha_i$ or $\beta_i$.
		\item $v_{i+1}$ is neighbor of either $\alpha_i$ or $\beta_i$.
	\end{enumerate}
	We will prove this by induction. First observe that $(1)$ holds for $i = 0$, since $v_0 = p_1$.
	Assume the strong induction hypothesis that both $(1),(2)$ holds for all $j < i$, and $(1)$ holds for $i$. We will show that this implies $(2)$ holds for $i$ and $(1)$ holds for $i+1$, completing the induction.
	Because of the induction hypothesis, $P$ must have visited all the vertices between $p_1$ and $\alpha_1, \beta_1$ in $\set{v_0, \ldots, v_i}$, since the set of visited vertices forms a contiguous segment on the circle.
	Suppose that $v_{i+1}$ is not a neighbor of either $\alpha_i$ or $\beta_i$. Then there is a vertex $v$ such that $v$ and $q_1$ are on the opposite sides of line $\overleftrightarrow{v_i v_{i+1}}$. Since $Q$ must visit $v$ before visiting $q_1$, it must intersect the line $\overleftrightarrow{v_i v_{i+1}}$.
	Since the segment $v_i v_{i+1}$ completely partitions the convex hull of $O_2$ into two parts, any path from $v$ to $q_1$ through the convex hull of $O_1$ must intersect $v_i v_{i+1}$, contradicting the planarity of the shortest path. This implies $(2)$.
	Further, since all the points between $\alpha_i$ and $\beta_i$ are already visited, $v_{i+1}$ is outside this segment, which implies that $v_{i+1}$ is either $\alpha_{i+1}$ or $\beta_{i+1}$.

	This proves the claim. The path $P$ must connect $H_1 \setminus \set{p_1, q_1}$ and $H_2 \setminus \set{p_1, q_1}$, and hence it crosses $\overleftrightarrow{p_1 q_1}$ at least once. Suppose it crosses the segment more than once.
	Let $v_a v_{a+1}$ and $v_b v_{b+1}$ be the two segments with least indices $a,b$ which cross $p_1 q_1$. Then $v_{a+1}$ is a neighbor of $p_1$ and $v_{b+1}$ is a neighbor of $v_a$. Let $p_2$ be neighbor of $p_1$ other than $v_{a+1}$. Note that $p_2$ is between $p_1$ and $v_a$.
	Consider the path $Q_1 = \overline{p_1 v_b}\,\overline{p_2v_{b+1}}$. We claim that this is shorter than the path	$Q_2 = \overline{p_1v_a}\,\overline{v_{a+1}v_b}v_{b+1}$, where $\overline{xy}$ denotes the path covering all the points between $x$ and $y$ which are on the same side of $\overleftrightarrow{p_1q_1}$ as $x,y$.
	Note that $\ell (\overline{p_1 v_a}) = \overline{p_2 v_{b+1}}$ and $\ell (\overline{v_{a+1} v_b}) + \ell (p_1 v_{a+1}) = \ell (\overline{p_1 v_b})$. Therefore, it suffices to show that 
	\[ \ell (v_a v_{a+1}) + \ell (v_b + v_{b+1}) - \ell (p_1 v_{a+1}) - \ell (v_b p_2) \ge 0 \]
	Let $\angle p_1 O v_{b+1} = \alpha \tfrac{2\pi}{k}$ and $\angle p_1 O v_b = \beta \tfrac{2\pi}{k}$ where $O$ is center of $O_1$. Then we can express all the lengths in terms of sines to get 
	\begin{align*}
			& \ell (v_a v_{a+1}) + \ell (v_b + v_{b+1}) - \ell (p_1 v_{a+1}) - \ell (v_b p_2) \\
		= & 2 \sin\brac{\frac{\alpha}{2} \cdot \frac{2\pi}{k}}
		+ 2 \sin\brac{\frac{\alpha + \beta}{2} \cdot \frac{2\pi}{k}}
		- 2 \sin\brac{\frac{\beta + 1}{2} \cdot \frac{2\pi}{k}}
		- 2 \sin\brac{\frac{1}{2} \cdot \frac{2\pi}{k}} \\
		= & 4 \sin\brac{\frac{2 \alpha + \beta}{2k} \pi} \cos\brac{\frac{\beta}{2k} \pi}
		- 4 \sin\brac{\frac{\beta + 2}{2k} \pi} \cos\brac{\frac{\beta}{2k} \pi} \\
		= & 4 \cos\brac{\frac{\beta}{2k} \pi}
		\brac{\sin\brac{\frac{2 \alpha + \beta}{2k} \pi} - \sin\brac{\frac{\beta + 2}{2k} \pi}} \\
		= & 8 \cos\brac{\frac{\beta}{2k} \pi} \cos\brac{\frac{\alpha + \beta + 1}{2k} \pi} \sin\brac{\frac{\alpha - 1}{2k} \pi} 
	\end{align*}
	Since $ 0 \le \beta \le \alpha + \beta + 1 \le k$, and $\alpha \ge 1$, all the angles in the expression above are between $0$ and $\tfrac{\pi}{2}$, which proves that this expression is always positive implying that $Q_1$ is shorter than $Q_2$.
	We can now replace portion of $Q$ corresponding to $Q_2$ by $Q_1$ to get a shorter path if $Q$ crossed $\overleftrightarrow{p_1 q_1}$ more than once. Hence, any optimal Hamiltonian path $Q$ must cross $\overleftrightarrow{p_1 q_1}$ exactly once.
	Therefore, $Q$ must look like $Q = \overline{p_1q_2}\,\overline{p_2q_1}$, where $q_2$ is a neighbor of $q_1$ that is on the opposite side of $p_2$. The points $p_1 p_2 q_1 q_2$ form a cyclic trapezoid, with $p_1 q_1$ and $p_2 q_2$ as diagonals. Therefore,
	\[ \ell (Q) \ge \dist(p_1 q_1) + 2 \pi - \frac{4\pi}{k} \]
	Note that we are missing the trivial case when $p_1$ and $q_1$ are adjacent, which can be verified to give the exact same bound. 
	
	Hence, portion of path $P$ in between two gap vertices has length $\ell (g_1 p_1) + \ell (p_1 q_1) + 2\pi - \frac{4\pi}{k} + \ell (q_1 g_2)$, which is at least $4 + 2 \pi - \tfrac{4\pi}{k}$. Combined with the cost of the path outside two gap vertices, which is at least $8 \pi + 4 - \tfrac{8\pi}{k}$, we get the lower bound
	\[ \ell (P) \ge 10 \pi + 8 - \frac{12\pi}{k} - O(k^{-2}) \]
	as required. Error term of $O(k^{-2})$ appears from approximating small chords of circle by the arc-lengths. 
\end{proof}
In particular, the lemma above gives the lower bound 
\[ \ell (P) \ge 10 \pi + 8 - O\brac{\frac{1}{k}} \]
	as required in \Cref{l:s_k_bound}. 

\section{Probability bounds for Observation \ref{obs:smiley_face_lemma}}
In this section we will provide precise bounds for constant $C_{\varepsilon,D}^S$ defined in \Cref{obs:smiley_face_lemma}. More precisely, we will prove the following version of \Cref{obs:smiley_face_lemma}:
\begin{lemma}
	\label{l:smiley_face_lemma}
	Let $d$ be a fixed integer. Let $\set{Y_1, Y_2, \ldots}$ be a sequence of points drawn uniformly at random from $[0,t]^d$ and $\YYY_n = \set{Y_1, \ldots, Y_n}$, where $t = n^{1/d}$. Given any finite point set $S \subseteq \R^d$ with $k$ points, any $\varepsilon > 0$ and any constant $D > 0$ such that 
	\begin{enumerate}[nosep]
	    \item $\varepsilon$ is smaller than distance between any two points in $S$; and
	    \item $D$ is larger than the diameter of $S$
	    \item $\exp(O(k \log(1 / \varepsilon))) = o(n)$
	\end{enumerate}
	$\YYY_n$ contains at least $C_{\varepsilon, D}^k n$ many $(\varepsilon, D)$-copies of $S$ with probablity $1 - o(1)$, where
	\[ C_{\varepsilon, D}^S = \exp \brac{-O(k \log (1/\varepsilon))} \]
	where $O$-notation hides constants dependent on $d$ and $D$.
\end{lemma}
\begin{proof}
	Divide $[0,t]^d$ into boxes of of side length $3D$. Let $B$ denote one such box. Consider a copy of $S$ centered at center of the box $B$. Let $s_1, \ldots, s_k$ be points in $S$. For any $j$, the probability that $Y_j$ at most $\varepsilon$ distance from $s_i$ is $\tfrac{V_d(\varepsilon)}{n}$ where $V_d(R)$ denotes volume of a sphere of radius $R$ in $\R^d$.
	Given a sequence of points $j_1, \ldots, j_k$, the probability that $Y_{j_i}$ is $\varepsilon$-close to $s_i$ for all $i$, and there are no other points inside $B$ is given by
	\[ \brac{\frac{V_d(\varepsilon)}{n}}^k \brac{1 - \frac{(3D)^d}{n}}^{n-k} \]
	The number of choices for the sequence $j_i$ is exactly
	\[ \frac{n!}{(n-k)!} \ge \brac{1 - \frac{k}{n}}^k n^k \]
	Since the events corresponding to all the sequences are disjoint, we can simply add these probabilities! Recall that $\log V_d(\varepsilon) = O(- d \log d + d \log \varepsilon)$, and that $1 - x \ge e^{-x/(1-x)} \ge e^{-2x}$ for $x \le \tfrac{1}{2}$. Using these two identities, and the two probability bounds above, we get a lower bound on probability that the box $B$ contains an $(\varepsilon,D)$-copy of $S$:
	\[ \exp\brac{-O\brac{dk \log d - dk \log \varepsilon + (n-k) \frac{(3D)^d}{n} + \frac{k^2}{n}}} = \exp(-O(k \log (1/\varepsilon))) \]
	where $O$ hides constants dependent on the dimension $d$.

	There are $\frac{n}{(3D)^d}$ such boxes $B$. Let these be denoted by $B_1, \ldots, B_s$. Let $\chi_i$ be the indicator random variable for box $B_i$ containing an $(\varepsilon,D)$-copy of $S$. Let $\chi = \sum_i \chi_i$. Since changing $Y_i$ for any particular $i$ only changes value of $\chi$ by $2$, we can apply McDiarmid's Inequality:
	
	\begin{lemma}[McDiarmid's Inequality]
	    Suppose a function $f: \mathcal{Z}_1 \times \cdots \mathcal{Z}_n \to \R$ satisfies that for all $i$
	    \[ \sup_{z_i' \in \mathcal{Z}_i} \abs[\big]{f(z_1, \ldots, z_{i-1}, z_i, z_{i+1}, \ldots , z_n) - f(z_1, \ldots, z_{i-1}, z_{i'}, z_{i+1}, \ldots, z_n)} \le c_i \]
	    then for independent random variables $Z_i \sim \mathcal{Z}_i$,
	    \[ \pr[\Big]{f(Z_1, \ldots, Z_n) - \ex{f(Z_1, \ldots, Z_n} \le -\varepsilon} \le \exp\brac{- \frac{2\varepsilon^2}{\sum_{i=1}^n c_i^2}} \]
	\end{lemma}
	
	Note that changing a single point $y \in \YYY_n$ changes $\chi$ by at most $2$. Therefore, we can use McDiarmid's Inequality with $c_i = 2$ for all $i$ to get that $\YYY_n$ contains at least 
	\[ \frac{1}{2} \exp(-O(k \log (1/\varepsilon)) - d \log (3D)) n = \exp(- O(k \log (1/\varepsilon))) n \]
	many $(\varepsilon,D)$-copies of $S$, with probability $1 - \exp(-\exp(-O(k \log (1 / \varepsilon)))n) = 1 - o(1)$ provided that $\exp(O(k \log ( 1 / \varepsilon))) = o(n)$.
\end{proof}

\end{document}